\documentclass{article}

\usepackage{amsmath}
\usepackage{amssymb}
\usepackage{amsthm}

\usepackage{enumerate}
\usepackage{microtype}
\usepackage{bussproofs}
\usepackage{float}
\floatstyle{boxed}
\restylefloat{table}
\restylefloat{figure}

\usepackage{pgf,tikz}
\usetikzlibrary{shapes,calc,matrix,
arrows,automata,positioning,backgrounds,fit,decorations.pathreplacing}

\newtheorem{definition}{Definition}
\newtheorem{theorem}{Theorem}
\newtheorem{proposition}[theorem]{Proposition}

\newcommand{\cand}{\text{ and }}

\makeatletter
\newcommand{\raisemath}[1]{\mathpalette{\raisem@th{#1}}}
\newcommand{\raisem@th}[3]{\raisebox{#1}{$#2#3$}}
\makeatother

\newcommand{\bra}{\langle}
\newcommand{\ket}{\rangle}

\newcommand{\bracc}[1]{\ensuremath{\left[\!\left[ {#1} \right]\!\right]}}

\newcommand{\fall}[1]{{\forall\,{#1},\ }}

\newcommand{\mb}[1]{{\bf #1}}
\newcommand{\mc}[1]{{\mathcal{#1}}}
\newcommand{\mf}[1]{{\mathfrak #1}}

\newcommand{\sas}{\mathbin{\&}}

\newcommand{\Imp}{\mathop{\mathtt{Imp}}}

\newcommand{\Mes}{\mathop{\mathrm{Mes}}\nolimits}

\newcommand{\Id}{\mathop{\mathrm{Id}}}
\newcommand{\Vect}{\mathop{\mathrm{Vect}}}
\newcommand{\Slice}{\mathop{\mathrm{Slice}}}

\newcommand{\comp}{\mathrel{\mathrm{C}}}
\newcommand{\contains}{\mathrel{\subseteq_f}}
\makeatletter
\newcommand{\atfter}{\mathrel{\ @\ }}
\makeatother
\newcommand{\blank}{\,\rule{5pt}{.5pt}\,}

\newtheorem{corollary}{Corollary}

\makeatletter
\newcommand{\verif@es}{\blacktriangleright}
\newcommand{\verifies}{\mathrel{\verif@es}}
\newcommand{\verifiestar}{\mathrel{\verif@es^{\!\scriptstyle \#}}}
\newcommand{\verifiesm}[1]{\mathrel{\verif@es_{{\scriptstyle \mathfrak {#1}}}}}
\newcommand{\verifiestarm}[1]{\mathrel{\verif@es^{\!\scriptstyle \#}_{{\!\scriptstyle \mathfrak {#1}}}}}
\makeatother

\newcommand{\transition}[1]{\mathbin{\,\xrightarrow{\ #1\ }\,}}

\begin{document}

\title{On the Possibility of Quantum Circuits \\ Part~I: the Epistemic Level}

\author{Olivier Brunet \\ \texttt{olivier.brunet at normalesup.org}}

\maketitle

\begin{abstract}
We present a formulation of quantum circuits where the focus is set on whether a given circuit (made of unitary operators and projective measurements with definite outcomes) does reflect an actually realizable physical experiment. In order to do this, we introduce \emph{verifications statements} which are purely epistemic assertions indicating whether a outcome is possible at some point and develop our formalism which, in the end, consists in a set of logical rules about verification statements, as summarized in figure~\ref{fig:last_logic} on page~\pageref{fig:last_logic}. Finally, we argue that our formalism provides a Lorentz-invariant realistic formulation of quantum circuits and illustrate this by considering a circuit corresponding to Hardy's paradox and showing how our formalism prevents making contradictory assertions regarding our knowledge about the circuit.
\end{abstract}







\section{Introduction}


In this article, we will introduce a formulation of a fragment of quantum mechanics (corresponding to quantum circuits) based on a possibilistic~\cite{Fritz:Possibilistic} rather than probabilistic approach: the question we want to investigate is the definition of a general characterization of whether a given circuit corresponds to a actually feasable physical experiment. For instance, consider the circuit depicted in figure~\ref{fig:very_first} where a single particle $A$ is measured twice in a row, with successive outcomes $|0\ket$ (or, more precisely, the subspace $[0]$ spanned by $|0\ket$) and $[1]$. If the two measurements are projective, this circuit does not reflect the outcomes of an actual experiment. In particular, at $A_2$, after the first measurement occured, following the Born rule, it is not possible to obtain any outcome orthogonal to~$[0]$, which we will denote $A_2 \verifies [0]$ (we say that the circuit verifies $[0]$ at $A_2$, or even that $A_2$ verifies $[0]$).

Another example is illustrated in figure~\ref{fig:first}: two particles, $A$ and $B$ are first measured, with outcomes $[0]$ and $[1]$. Then, they are applied a controlled-not gate, and measured again, with the same outcomes. If one reasons in terms of quantum states, after the first measurements, at $A_2$ and $B_2$, the particles are respectively in states $|0\ket$ and $|1\ket$ so that after the CNot gate is applied, they are now both in state $|1\ket$, so that the obtention of outcome $[0]$ when measuring particle $A$ at $A_3$ is impossible. In fact, measuring both particles jointly, any outcome orthogonal to $[0] \otimes [1]$ is impossible, which we write $A_3, B_3 \verifies [0] \otimes [1]$. However, instead of relying on the delicate and elusive notion of quantum state, we will rather base our discussion on projective measurement outcomes.

In the following, after a brief presentaton of the formalization of quantum circuits we will use and some related notions, we will make a few assumptions about the way projective measurements act, how outcomes can follow each other -- so that a measurement outcome induces a prediction about a potential future outcome --, and how the application of unitary outcomes modifies these predictions. We will then define \emph{verification statements}, which correspond to one particular type of prediction, and our assumptions about the behaviour of projective measurements will lead us to the definition of a set of logical rules about verifications statements (the final version of which is presented in figure~\ref{fig:last_logic} on page~\pageref{fig:last_logic}). Finally, we will argue that the obtained formalism provides a Lorentz-invariant realistic formulation of quantum mechanics (or, at least, of the fragment corresponding to quantum circuits), and we will illustrate this by studying the modelization of Hardy's paradox in our approach and showing how some arguments forbidding any Lorentz-invariant interpretation of quantum mechanics are not valid therein.

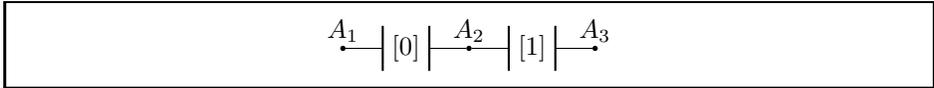
\begin{figure}
\begin{centering}
\begin{tikzpicture}[point/.style={circle,inner sep=0pt,outer sep = 0pt, minimum size=2pt,fill=black}]
	\matrix[row sep=4mm,column sep=5mm]{
		\node (a0) [point] {} ;
		&
		\node (a1) [rectangle] {$[0]$} ;
		&
		\node (a2) [point] {} ;
		&
		\node (a3) [rectangle] {$[1]$} ;
		&
		\node (a4) [point] {} ; \\
 	};
	\node [above = -2pt of a0] {$A_1$} ;
	\node [above = -2pt of a2] {$A_2$} ;
	\node [above = -2pt of a4] {$A_3$} ;
	\draw (a0) -- (a1) -- (a2) -- (a3) -- (a4) ;
	\draw [thick] (a1.north west) -- (a1.south west) ;
	\draw [thick] (a1.north east) -- (a1.south east) ;
	\draw [thick] (a3.north west) -- (a3.south west) ;
	\draw [thick] (a3.north east) -- (a3.south east) ;
\end{tikzpicture} \\
\end{centering}
\caption{An impossible circuit} \label{fig:very_first}
\end{figure}

\begin{figure}
\begin{centering}
\begin{tikzpicture}[point/.style={circle,inner sep=0pt,outer sep = 0pt, minimum size=2pt,fill=black}]
	\matrix[row sep=4mm,column sep=5mm]{
		\node (a0) [point] {} ;
		&
		\node (a1) [rectangle] {$[0]$} ;
		&
		\node (a2) [point] {} ;
		&
		\node [inner sep = -10pt, outer sep = 0pt] (a3) {$\bigoplus$} ;
		&
		\node (a4) [point] {} ;
		&
		\node (a5) [rectangle] {$[0]$} ; 
		&
		\node (a6) [point] {} ;
		\\
		\node (b0) [point] {} ;
		&
		\node (b1) [rectangle] {$[1]$} ;
		&
		\node (b2) [point] {} ;
		&
		\node [circle, inner sep=0pt,outer sep = 0pt, minimum size=4pt, fill=black] (b3) {};
		&
		\node (b4) [point] {} ;
		&
		\node (b5) [rectangle] {$[1]$} ;
		&
		\node (b6) [point] {} ;
		\\
 	};
	\node [above = -2pt of a0] {$A_1$} ;
	\node [above = -2pt of a2] {$A_2$} ;
	\node [above = -2pt of a4] {$A_3$} ;
	\node [above = -2pt of a6] {$A_4$} ;
	\node [above = -2pt of b0] {$B_1$} ;
	\node [above = -2pt of b2] {$B_2$} ;
	\node [above = -2pt of b4] {$B_3$} ;
	\node [above = -2pt of b6] {$B_4$} ;
	\draw (a0) -- (a1) -- (a2) -- (a3) -- (a4) -- (a5) -- (a6) ;
	\draw (b0) -- (b1) -- (b2) -- (b3) -- (b4) -- (b5) -- (b6) ;
	\draw (b3) -- (a3) ;
	\draw [thick] (a1.north west) -- (a1.south west) ;
	\draw [thick] (a1.north east) -- (a1.south east) ;
	\draw [thick] (a5.north west) -- (a5.south west) ;
	\draw [thick] (a5.north east) -- (a5.south east) ;
	\draw [thick] (b1.north west) -- (b1.south west) ;
	\draw [thick] (b1.north east) -- (b1.south east) ;
	\draw [thick] (b5.north west) -- (b5.south west) ;
	\draw [thick] (b5.north east) -- (b5.south east) ;
\end{tikzpicture}

\end{centering}
\caption{Another impossible circuit} \label{fig:first}
\end{figure}
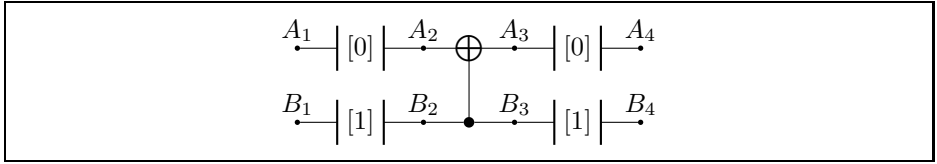

\section{Quantum Circuits}

Let us first define quantum circuits formally. They are acyclic oriented graphs with two types of nodes:
\begin{itemize}
	\item s-nodes (as for \emph{system}) which represent parts of a quantum system at a given stage of the circuit,
	\item o-nodes (as for \emph{operation}) which represent quantum operations applied to the system. Basically, we will consider two types of operations: projective measurements and unitary transformations.
\end{itemize}
Regarding measurements, the consideration of projective measurements only cannot be seen as a limitation: as we will be able to deal with composite systems, it will be possible to simulate POVMs as a consequence of Naymark's theorem \cite{NielsenChung2000Book,Peres2002:QuantumTheory}.

Quantum circuits are bipartite: any arrow must connect two nodes of different types.

Each s-node $s$ has a dimension $d(s)$ corresponding to the dimension of the Hilbert space used to model observables applicable to $s$. Moreover, each s-node may have \emph{at most} one incoming arrow and one outgoing one.

Formally, each o-node $U$ is defined by the following informations:
\begin{itemize}
\item an ordered list $(i_1, \ldots, i_p)$ specifying the number and dimensions of its incoming nodes,
\item another similar list $(o_1, \ldots, o_q)$ for the outgoing nodes,
\item a linear operator $[\![U]\!]$ from $\mb C^{\prod i_k}$ to $\mb C^{\prod o_k}$, which is either unitary or an orthogonal projection.
\end{itemize}
However, in practice, we don't to use such a heavy machinery directly. In the following, circuits will be described graphically, and the differents incoming and outgoing edges will be easily distinguishable (the only relevant type of o-node being C-Not gates). It has to keep in mind, though, that such a graphical depiction is only a handy way to describe o-nodes in such a way that the different incoming and outgoing edges can be distinguished.



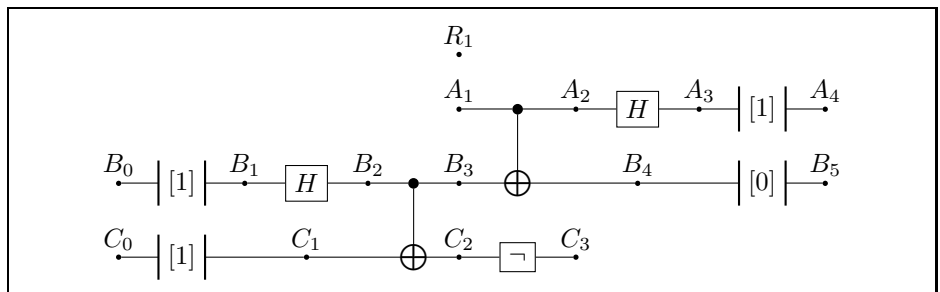
\begin{figure}
\begin{centering}
\begin{tikzpicture}[point/.style={circle,inner sep=0pt,outer sep = 0pt, minimum size=2pt,fill=black}]
	\matrix[row sep=4mm,column sep=5mm]{
		& & & & & & \node (r0) [point] {} ; \\
		& & & & & & \node (a0) [point] {} ;
		&
		\node (a1) [circle, inner sep=0pt,outer sep = 0pt, minimum size=4pt, fill=black] {} ;
		&
		\node (a2) [point] {} ;
		&
		\node[rectangle, draw] (a3) {$H$} ;
		&
		\node (a4) [point] {} ;
		&
		\node (a5) [rectangle] {[1]} ;
		&
		\node (a6) [point] {} ;
		\\
		\node(bz) [point] {} ;
		&
		\node (b0) [rectangle] {$[1]$};
		&
		\node (b1) [point] {};
		&
		\node[rectangle, draw] (b2) {$H$} ;
		&
		\node (b3) [point] {};
		&
		\node [circle, inner sep=0pt,outer sep = 0pt, minimum size=4pt, fill=black] (b4) {};
		&
		\node (b5) [point] {};
		&
		\node [inner sep = -10pt, outer sep = 0pt] (b6) {$\bigoplus$} ;
		& &
		\node (b7) [point] {};
		& &
		\node (b8) [rectangle] {[0]} ;
		&
		\node (b9) [point] {} ;
		\\
		\node (cz) [point] {} ;
		&
		\node (c0) [rectangle] {$[1]$} ;
		& &
		\node (c1) [point] {};
		& &
		\node [inner sep = -10pt, outer sep = 0pt] (c2) {$\bigoplus$} ;
		&
		\node (c3) [point] {};
		&
		\node [rectangle, draw] (c4) {$\neg$} ;
		&
		\node (c5) [point] {};
		\\
	};
\node [above = -2pt of r0] {$R_1$} ;
\node [above = -2pt of a0] {$A_1$} ;
\node [above = -2pt of a2] {$A_2$} ;
\node [above = -2pt of a4] {$A_3$} ;
\node [above = -2pt of a6] {$A_4$} ;
\node [above = -2pt of bz] {$B_0$} ;
\node [above = -2pt of b1] {$B_1$} ;
\node [above = -2pt of b3] {$B_2$} ;
\node [above = -2pt of b5] {$B_3$} ;
\node [above = -2pt of b7] {$B_4$} ;
\node [above = -2pt of b9] {$B_5$} ;
\node [above = -2pt of cz] {$C_0$} ;
\node [above = -2pt of c1] {$C_1$} ;
\node [above = -2pt of c3] {$C_2$} ;
\node [above = -2pt of c5] {$C_3$} ;
\draw (a0) -- (a1) -- (a2) -- (a3) -- (a4) -- (a5) -- (a6) ;
\draw [thick] (a5.north west) -- (a5.south west) ;
\draw [thick] (a5.north east) -- (a5.south east) ;
\draw (bz) -- (b0) -- (b1) -- (b2) -- (b3) -- (b4) -- (b5) -- (b6) -- (b7) -- (b8) -- (b9);
\draw [thick] (b0.north west) -- (b0.south west) ;
\draw [thick] (b0.north east) -- (b0.south east) ;
\draw [thick] (b8.north west) -- (b8.south west) ;
\draw [thick] (b8.north east) -- (b8.south east) ;
\draw (cz) -- (c0) -- (c1) -- (c2) -- (c3) -- (c4) -- (c5) ;
\draw [thick] (c0.north west) -- (c0.south west) ;
\draw [thick] (c0.north east) -- (c0.south east) ;
\draw (c2) -- (b4) ;
\draw (b6) -- (a1) ;
\end{tikzpicture}

\end{centering}
\caption{$\mf T$, a larger circuit} \label{fig:T_uncut}
\end{figure}

\ 

The acyclic structure of a quantum circuit induces a partial order between its nodes:
\begin{definition}[Order relation]
A node $a$ is \emph{in the past} of a node $b$ in $\mf C$ is there is a finite sequence of arrows from $a$ to $b$. In that case, we write
$$ a \transition {\mf C} b $$
\end{definition}

\begin{definition}[Full subgraph]
Given two quantum circuits $\mf C$ and $\mf C'$, we say that $\mf C$ is a \emph{full subgraph} of $\mf C'$ if $\mf C \subseteq \mf C'$ and, moreover, $\mf C$ preserves the space time structure of $\mf C'$, that~is
$$ \fall {a, b \in \mf C} a \transition {\mf C} b \iff a \transition {\mf C'} b $$
We denote this $\mf C \contains \mf C'$.
\end{definition}

\begin{proposition}
Being a full subgraph is a partial order among quantum circuits, so that for all $\mf C$, $\mf C'$ and $\mf C''$,
$$ \mf C \contains \mf C \qquad \bigl(\mf C \contains \mf C' \cand \mf C' \contains \mf C''\bigr) \implies \mf C \contains \mf C'' $$
\end{proposition}


An important element of our formalism will rely on what we will define as slices: In a quantum circuit, an s-node can be seen as corresponding to an spacetime event, and it will be interesting to considered simultaneously several spacelike separated events, where being spacelike separated corresponds to the fact that, in the circuit, they are not comparable with regards to the previously defined order relation.

\begin{definition}[Slice]
Given a circuit $\mf C$, a \emph{slice} $\Gamma$ of $\mf C$ is a ordered set of mutually uncomparable s-nodes of $\mf C$. Let $\Slice(\mf C)$ denote the set of slices of $\mf C$.

The \emph{support} of a slice $\Gamma = [s_1, \ldots, s_n]$ is the set made of its s-nodes:
$$ \Gamma_{\{\}} = \bigl\{s_1, \ldots, s_n\bigr\} $$
The \emph{dimension} of a slice $\Gamma = [s_1, \ldots, s_n]$ is defined as the product of the dimension of its nodes:
$$ d(\Gamma) = d(s_1) \times \ \cdots \ \times d(s_n) $$
\end{definition}

\begin{proposition}
If $\mf C \contains \mf C'$, then $\Slice(\mf C) \subseteq \Slice(\mf C')$.
\end{proposition}

Another important idea in our approach will be to cut a circuit along a given slice, by removing its future, and possibly replacing it by a measurement o-node.

\begin{definition}[Cutting a circuit along a slice]
Given a graph $\mf C$ and a slice $\Gamma \in \Slice(\mf C)$, we define $\mf C|_{\Gamma}$ as the graph obtained from $\mf C$~by removing all the nodes (both s-nodes and o-nodes) in the strict future of~$\Gamma$.
\end{definition}

An example of such a cut, namely the cut of $\mf T$ along $\{B_3, C_2\}$ is depicted in figure~\ref{fig:T_cut}. In figure~\ref{fig:cut_measure}, one has the same circuit with a measurement o-node added.
Finally, let us remark that cutting a circuit $\mf C$ along a slice provides a full subgraph of $\mf C$:

\begin{proposition} \label{prop:cutting_is_full}
If $\Gamma \in \Slice(\mf C)$, then $\mf C|_\Gamma \contains \mf C$.
\end{proposition}

\begin{figure}
\begin{centering}
\begin{tikzpicture}[point/.style={circle,inner sep=0pt,outer sep = 0pt, minimum size=2pt,fill=black}]
	\matrix[row sep=4mm,column sep=5mm]{
		& & & & & & \node (r0) [point] {} ; \\
		& & & & & & \node (a0) [point] {} ;
		&
		\node (a1) [circle, inner sep=0pt,outer sep = 0pt, minimum size=4pt, fill=black, black!25] {} ;
		&
		\node (a2) [point, black!25] {} ;
		&
		\node[rectangle, draw, black!25] (a3) {$H$} ;
		&
		\node (a4) [point, black!25] {} ;
		&
		\node (a5) [rectangle, black!25] {[1]} ;
		&
		\node (a6) [point, black!25] {} ;
		\\
		\node (bz) [point] {};
		&
		\node (b0) [rectangle] {$[1]$};
		&
		\node (b1) [point] {};
		&
		\node[rectangle, draw] (b2) {$H$} ;
		&
		\node (b3) [point] {};
		&
		\node [circle, inner sep=0pt,outer sep = 0pt, minimum size=4pt, fill=black] (b4) {};
		&
		\node (b5) [point] {};
		&
		\node [inner sep = -10pt, outer sep = 0pt, black!25] (b6) {$\bigoplus$} ;
		& &
		\node (b7) [point, black!25] {};
		& &
		\node (b8) [rectangle, black!25] {$[0]$} ;
		&
		\node (b9) [point, black!25] {} ;
		\\
		\node (cz) [point] {} ;
		&
		\node (c0) [rectangle] {$[1]$} ;
		& &
		\node (c1) [point] {};
		& &
		\node [inner sep = -10pt, outer sep = 0pt] (c2) {$\bigoplus$} ;
		&
		\node (c3) [point] {};
		&
		\node [rectangle, draw, black!25] (c4) {$\neg$} ;
		&
		\node (c5) [point, black!25] {};
		\\
	};
\node [above = -2pt of r0] {$R_1$} ;
\node [above = -2pt of a0] {$A_1$} ;
\node [above = -2pt of a2, black!25] {$A_2$} ;
\node [above = -2pt of a4, black!25] {$A_3$} ;
\node [above = -2pt of a6, black!25] {$A_4$} ;
\node [above = -2pt of bz] {$B_0$} ;
\node [above = -2pt of b1] {$B_1$} ;
\node [above = -2pt of b3] {$B_2$} ;
\node [above = -2pt of b5] {$B_3$} ;
\node [above = -2pt of b7, black!25] {$B_4$} ;
\node [above = -2pt of b9, black!25] {$B_5$} ;
\node [above = -2pt of cz] {$C_0$} ;
\node [above = -2pt of c1] {$C_1$} ;
\node [above = -2pt of c3] {$C_2$} ;
\node [above = -2pt of c5, black!25] {$C_3$} ;
\draw [black!25] (a0) -- (a1) -- (a2) ;
\draw [black!25] (a2) -- (a3) -- (a4) -- (a5) -- (a6) ;
\draw [thick, black!25] (a5.north west) -- (a5.south west) ;
\draw [thick, black!25] (a5.north east) -- (a5.south east) ;
\draw (bz) -- (b0) -- (b1) -- (b2) -- (b3) -- (b4) -- (b5) ;
\draw [black!25] (b5) -- (b6) -- (b7) -- (b8) -- (b9);
\draw [thick] (b0.north west) -- (b0.south west) ;
\draw [thick] (b0.north east) -- (b0.south east) ;
\draw [thick, black!25] (b8.north west) -- (b8.south west) ;
\draw [thick, black!25] (b8.north east) -- (b8.south east) ;
\draw (cz) -- (c0) -- (c1) -- (c2) -- (c3) ;
\draw [black!25] (c3) -- (c4) -- (c5) ;
\draw [thick] (c0.north west) -- (c0.south west) ;
\draw [thick] (c0.north east) -- (c0.south east) ;
\draw (c2) -- (b4) ;
\draw [black!25] (b6) -- (a1) ;
\end{tikzpicture} \\
\end{centering}
\caption{Removing the strict future of $\{B_3, C_2\}$ in $\mf T$, yielding $\mf T|_{\{B_3, C_2\}}$} \label{fig:T_cut}
\end{figure}
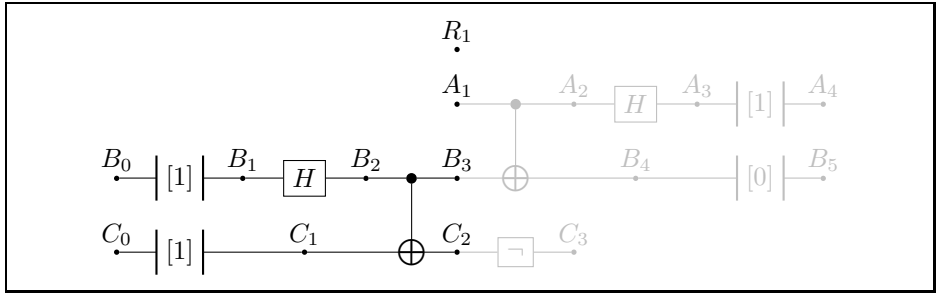









\begin{figure}
\begin{centering}
\begin{tikzpicture}[point/.style={circle,inner sep=0pt,outer sep = 0pt, minimum size=2pt,fill=black}]
	\matrix[row sep=5mm,column sep=5mm]{
		& & & & & & \node (r0) [point] {} ; \\
		& & & & & & \node (a0) [point] {} ; & \node (a1) [rectangle] {\vphantom{$[0]$}};\\
		\node (bz) [point] {} ;
		&
		\node (b0) [rectangle] {$[1]$};
		&
		\node (b1) [point] {};
		&
		\node[rectangle, draw] (b2) {$H$} ;
		&
		\node (b3) [point] {};
		&
		\node [circle, inner sep=0pt,outer sep = 0pt, minimum size=4pt, fill=black] (b4) {};
		&
		\node (b5) [point] {};
		&
		\node (b6) [rectangle] {[1]} ; ;
		&
		\node (b7) [point, black] {}; \\
		\node (cz) [point] {} ;
		&
		\node (c0) [rectangle] {$[1]$} ;
		& &
		\node (c1) [point] {};
		& &
		\node [inner sep = -10pt, outer sep = 0pt] (c2) {$\bigoplus$} ;
		&
		\node (c3) [point] {};
		&
		\node [rectangle] (c4) {$[0]$} ;
		&
		\node (c5) [point] {};
		\\
	};
\node [above = -2pt of r0] {$R_1$} ;
\node [above = -2pt of a0] {$A_1$} ;
\node [above = -2pt of bz] {$B_0$} ;
\node [above = -2pt of b1] {$B_1$} ;
\node [above = -2pt of b3] {$B_2$} ;
\node [above = -2pt of b5] {$B_3$} ;
\node [above = -2pt of b7] {$B_4$} ;
\node [above = -2pt of cz] {$C_0$} ;
\node [above = -2pt of c1] {$C_1$} ;
\node [above = -2pt of c3] {$C_2$} ;
\node [above = -2pt of c5] {$C_3$} ;
\draw (bz) -- (b0) -- (b1) -- (b2) -- (b3) -- (b4) -- (b5) -- (b6) -- (b7);
\draw [thick] (b0.north west) -- (b0.south west) ;
\draw [thick] (b0.north east) -- (b0.south east) ;
\draw [thick] (b6.north west) -- (c4.south west) ;
\draw [thick] (b6.north east) -- (c4.south east) ;
\draw (cz) -- (c0) -- (c1) -- (c2) -- (c3) -- (c4) -- (c5) ;
\draw [thick] (c0.north west) -- (c0.south west) ;
\draw [thick] (c0.north east) -- (c0.south east) ;
\draw (c2) -- (b4) ;
\end{tikzpicture} \\\end{centering}
\caption{$\mf T|_{\{B_3, C_2\}} \cup \bigl\{B_4, C_3 = \Mes_{[10]}(B_3, C_2)\bigr\}$} \label{fig:cut_measure}
\end{figure}

\section{Possible and Impossible Circuits}

Having defined the formalism for representing quantum circuits, let us now present some assumptions regarding whether a given quantum circuit is \emph{possible}. These assumptions will only relate to the obtention of measurement outcomes. In particular, no mention will be made of any notion of quantum state. Instead, the rules we shall enounce will correspond to some situations which, as we will assume, cannot correspond to an actual physical situation. If it is the case, if a quantum circuit~$\mf C$ can be shown to be impossible (with regards to our assumptions), we will denote
$$ \mf C \vdash \Imp $$

Obviously, it has to be kept in mind that these assumptions must be compatible with the standard formulation of quantum mechanics in order to make correct predictions.

\subsection{Non Contradiction}

Our first assumption is that, since we only consider projective measurements, it is not possible to obtain two orthogonal outcomes when measuring the same system twice in a row. Diagramatically, this means that any quantum circuit~$\mf C$ containing two consecutive projective measurements with orthogonal outcomes is impossible, as depicted in figure~\ref{fig:basic_impossibility}.


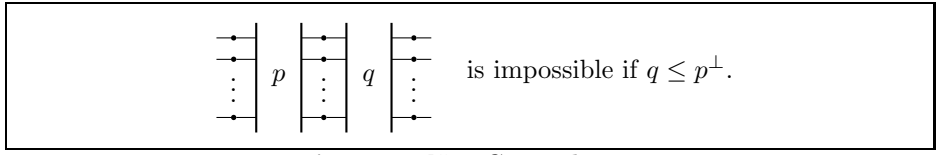
\begin{figure}[ht]
\hbox{} \hfill
\begin{tikzpicture}[baseline=(current bounding box.center),point/.style={circle,inner sep=0pt,outer sep = 0pt, minimum size=2pt,fill=black}]
	\matrix[row sep=-1mm]{
		\node (a0) {} ; &[2pt] \node (a1) [point] {}; &[4pt] \node[rectangle] (A1) {$\ \vphantom{a}\ \ $} ; &[4pt] \node (a2) [point] {}; &[4pt] \node[rectangle] (A2) {$\ \vphantom{a}\ \ $} ; &[4pt] \node (a4) [point] {}; &[2pt] \node (a5) {} ; \\
		\node (b0) {} ; & \node (b1) [point] {}; & \node[rectangle] (B1) {$\ \vphantom{a}\ \ $} ; & \node (b2) [point] {}; & \node[rectangle] (B2) {$\ \vphantom{a}\ \ $} ; & \node (b4) [point] {}; & \node (b5) {} ; \\[-5pt]
	 	& \node (c) {\raisebox{5pt}{$\vdots$}} ; & & \node (c) {\raisebox{5pt}{$\vdots$}} ; & & \node (c) {\raisebox{5pt}{$\vdots$}} ; \\[-5pt]
		\node (d0) {} ; & \node (d1) [point] {}; & \node[rectangle] (D1) {$\ \vphantom{a}\ \ $} ; & \node (d2) [point] {}; & \node[rectangle] (D2) {$\ \vphantom{a}\ \ $} ; & \node (d4) [point] {}; & \node (d5) {} ; \\[6pt]
	};
\draw (a0) -- (a1) -- (A1) -- (a2) -- (A2) -- (a4) -- (a5);
\draw (b0) -- (b1) -- (B1) -- (b2) -- (B2) -- (b4) -- (b5);
\draw (d0) -- (d1) -- (D1) -- (d2) -- (D2) -- (d4) -- (d5);
\draw [thick] (A1.north west) -- (D1.south west) ;
\draw [thick] (A1.north east) -- (D1.south east) ;
\node at ($(A1.north west)!.5!(D1.south east)$) {$p$} ;
\draw [thick] (A2.north west) -- (D2.south west) ;
\draw [thick] (A2.north east) -- (D2.south east) ;
\node at ($(A2.north west)!.5!(D2.south east)$) {$q$} ;
\end{tikzpicture} is impossible if $q \leq p^\bot$.%
\hfill \hbox{}
\caption{Non Contradiction}
\label{fig:basic_impossibility}
\end{figure}

In a more general way, if a circuit~$\mf C$ contains a measurement $\Gamma' = \Mes_p(\Gamma)$, then cutting~$\mf C$ at~$\Gamma'$ and inserting a measurement of the form $\Gamma'' = \Mes_q(\Gamma')$ leads to an impossible circuit if $q \leq p^\bot$.

Here, the only slice of interest is $\Gamma'$, so that in order to have lighter notations, we shall omit to explicitely name the other two, writing ``$\blank$'' instead. This way, the previous statement can be rephrased~as: if~$\mf C$ contains a measurement~$\Gamma = \Mes_p(\blank)$ (that is if $\Gamma = \Mes_p(\blank) \in \mf C$), then cutting it at $\Gamma$ and inserting~$\blank = \Mes_q(\Gamma)$ leads to an impossible circuit if $q \leq p^\bot$:
$$ \bigl(\Gamma = \Mes_p(\blank) \in \mf C \cand q \leq p^\bot\bigr) \implies \mf C|_\Gamma \cup \bigl\{\blank = \Mes_q(\Gamma)\bigr\} \vdash \Imp $$
In the following, we will write this as a logical rule:
\begin{prooftree}
\AxiomC{$\Gamma = \Mes_p(\blank) \in \mf C$}
\AxiomC{$q \leq p^\bot$}
\RightLabel{NC}
\BinaryInfC{$\mf C|_{\Gamma} \cup \bigl\{\blank = \Mes_{q}(\Gamma)\bigr\} \vdash \Imp$}
\end{prooftree}
where the top line corresponds the premises of the logical deduction (here, that~$\mf C$ contains $\Gamma = \Mes_p(\blank)$ and that $p$ and $q$ are orthogonal) and the bottom line to the conclusion which here states that cutting $\mf C$ at $\Gamma$ and inserting the measurement $\blank = \Mes_q(\Gamma)$ leads to an impossible circuit.

The name of the rule (here NC which stands for \emph{Non Contradiction}) is indicated on the right of the horizontal line.

\subsection{Considering Observables}

Our second assumption, divided in two parts, formalizes the behaviour of the measurement of observables. Let us define this notion, in the most general way, for an orthomodular lattice $L$:

\begin{definition}[Observable]
A (projective) observable of an orthomodular lattice $L$ is a finite subset $\{p_1, \ldots, p_n\}$ of $L$ such that~:
$$ \fall {i \in \bracc{1,n}} p_i \neq \bot, \quad \fall {i,j \in \bracc{1,n}} i \neq j \implies p_i \leq p_j^\bot \quad \hbox{and} \quad \bigvee_{i=1}^n p_i = \top $$
\end{definition}
In the following, we will use orthomodular lattices $L_n$ defined, for $n \in \mb N$, as the set of subspaces of $\mb C^n$. In particular, any outcome of an obervable applicable at slice $\Gamma$ will be a subspace of $\mb C^{d(\Gamma)}$, i.e.\ an element of $L_{d(\Gamma)}$. The top element $\top$ then corresponds to the whole vector space, while the bottom element $\bot$ is the nullspace.

\ 

We first remark that $\bot$ cannot be part of an observable, it is not a valid outcome. As such, any circuit containing an o-node of the form $\blank = \Mes_\bot(\blank)$ is impossible. We can write this as a rule the following~way (the name \emph{Top} will be clearer soon):
\begin{prooftree}
\AxiomC{$\blank = \Mes_\bot(\blank) \in \mf C$}
\RightLabel{Top}
\UnaryInfC{$\mf C \vdash \Imp$}
\end{prooftree}
or, equivalently, for any slice $\Gamma \in \Slice(\mf C)$:
\begin{prooftree}
\AxiomC{\vphantom{()}}
\UnaryInfC{$\mf C|_\Gamma \cup \{ \blank = \Mes_\bot(\Gamma) \} \vdash \Imp$}
\end{prooftree}

The other part of our assumption is that considering any slice $\Gamma$ of a quantum circuit $\mf C$, and any observable $\mc O = \{p_1, \ldots, p_n\}$ applicable at $\Gamma$, then if $\mf C$ is possible, then when measuring $\mf C$ at~$\Gamma$, at least one of the outcomes of $\mc O$ is possible. Considering the contraposition, if none of the outcomes of $\mc O$ is possible at~$\Gamma$ in $\mf C$, then $\mf C$ is impossible. As a rule, this can be expressed~as
\begin{prooftree}
\AxiomC{$\mf C|_{\Gamma} \cup \bigl\{\blank = \Mes_{p_1}(\Gamma)\bigr\} \vdash \Imp$}
\AxiomC{$\hspace{-5pt}\cdots\hspace{-5pt}$}
\AxiomC{$\mf C|_{\Gamma} \cup \bigl\{\blank = \Mes_{p_n}(\Gamma)\bigr\} \vdash \Imp$}
\RightLabel{Mes}
\TrinaryInfC{$\mf C \vdash \Imp$}
\end{prooftree}

In this rule, the pattern $\mf C|_\Gamma \vdash \{\blank = \Mes_p(\Gamma)\} \vdash \Imp$ appears several times, and it already appeared as the conclusion of the \emph{Non Contradiction} rule. This pattern will actually appear pervasively in our approach and this motivates the following definition:

\begin{definition}[Verification Statement]
Given a quantum circuit $\mf C$ and a slice $\Gamma$ of $\mf C$, we say that $\mf C$ \emph{verifies} $p$ at $\Gamma$ if and only~if:
$$ \mf C|_\Gamma \cup \bigl\{\blank = \Mes_{p^\bot}(\Gamma)\bigr\} \vdash \Imp $$
which we denote
$$ \mf C \vdash \Gamma \verifies p. $$
\end{definition}

With this definition, we can express the \emph{Mes} rule~as
\begin{prooftree}
\AxiomC{$\mf C \vdash \Gamma \verifies p_1^\bot$}
\AxiomC{$\hspace{-5pt}\cdots\hspace{-5pt}$}
\AxiomC{$\mf C \vdash \Gamma \verifies p_n^\bot$}
\RightLabel{Mes}
\TrinaryInfC{$\mf C \vdash \Imp$}
\end{prooftree}
Similarly, the previous \emph{Top} rule can be expressed~as
\begin{prooftree}
\AxiomC{\vphantom{()}}
\RightLabel{Top}
\UnaryInfC{$\mf C \vdash \Gamma \verifies \top$}
\end{prooftree}
in which form the name \emph{Top} becomes clear.
The \emph{Non Contradiction} rule becomes, in the special case where $q = p^\bot$:
\begin{prooftree}
\AxiomC{$\Gamma = \Mes_p(\blank) \in \mf C$}
\RightLabel{NC}
\UnaryInfC{$\mf C|_{\Gamma} \cup \bigl\{\blank = \Mes_{p^\bot}(\Gamma)\bigr\} \vdash \Imp$}
\end{prooftree}
and, in terms of verification statements, can be rewritten~as
\begin{prooftree}
\AxiomC{$\Gamma = \Mes_p(\blank) \in \mf C$}
\RightLabel{NC}
\UnaryInfC{$\mf C \vdash \Gamma \verifies p$}
\end{prooftree}

\subsection{Possibilistic Non-Contextuality}

The third assumption stems from the joint consideration of quantum mechanics and relativity. Following Aharonov and Albert \cite{Aharonov-Albert:IsUsualII}, consider a particle which may be located in any of three separate boxes $A$, $B$ and $C$ and suppose that one has enough knowledge to predict with certainty that the particle can neither be found in $B$ nor in~$C$. In that case, if an experimenter first opens boxes~$B$ and~$C$ (finding them empty), then the opening of box $A$ will lead to finding there our particle with certainty.

Suppose now that box $A$ is sufficiently far from the other two boxes, and consider a second reference frame where the opening of~$A$ happens \emph{before} that of~$B$ and~$C$. Obviously, in this reference frame, the measurement outcomes are the same and hence the particle will still be found in box~$A$.

But as~$A$ is opened before the other boxes, there is the possibility in this reference frame that, after the opening of box~$A$, the content of boxes $B$ and $C$ are modified, merged, exchanged, etc. In terms of observables, this means that an other observable can actually measured, with the restriction that $A$ must be one of its outcomes. In this situation, the particle remains to be found whatever happens later to boxes $B$ and $C$.

This leads to the assumption of \emph{possibilistic non-contextuality} which states that the certainty of an outcome is independant of which observable is actually measured (as long as the outcome remains a possible one), and a similar reasoning can be done regarding impossible outcomes.

\ 

Let us present two logical rules which follow from this assumption. First, suppose that it is impossible to obtain an outcome $p$ when measuring a circuit~$\mf C$ at~$\Gamma$, that~is
$$ \mf C|_\Gamma \cup \{ \blank = \Mes_p(\Gamma)\bigr\} \vdash \Imp $$
In that case, measuring observable $\{p, p^\bot\}$ at $\Gamma$ in $\mf C$ will yield outcome $p^\bot$ with certainty. For $q \leq p$, considering observable $\{q, p \wedge q^\bot, p^\bot\}$, from the certainty of outcome $p^\bot$, it follows that $q$ is not possible. We thus deduce the \emph{Order} rule:
\begin{prooftree}
\AxiomC{$\mf C|_\Gamma \cup \bigl\{\blank = \Mes_p(\Gamma)\bigr\} \vdash \Imp$}
\AxiomC{$q \leq p$}
\RightLabel{Ord}
\BinaryInfC{$\mf C|_\Gamma \cup \bigl\{\blank = \Mes_q(\Gamma)\bigr\} \vdash \Imp$}
\end{prooftree}
or, equivalently, using verification statements:
\begin{prooftree}
\AxiomC{$\mf C \vdash \Gamma \verifies p$}
\AxiomC{$p \leq q$}
\RightLabel{Ord}
\BinaryInfC{$\mf C \vdash \Gamma \verifies q$}
\end{prooftree}

Suppose now that two compatible outcomes $p$ and $q$ are assumed to be impossible at $\Gamma$, so that they both belong to a single boolean subalgebra of our orthomodular lattice, which we denote $p \comp q$. If we consider observable $\{p, p^\bot \wedge q, p^\bot \wedge q^\bot\}$, the impossibility of $p$ and of $p^\bot \wedge q$ (which follows, using the \emph{Ord} rule, from the impossibility of $q$) implies the certainty of $p^\bot \wedge q^\bot$.

Now, considering observable $\{p \vee q, p^\bot \wedge q^\bot\}$, the certainty of $p^\bot \wedge q^\bot$ implies the impossibility of $p \vee q$. We thus have derived the \emph{Compatible Meet} rule:
\begin{prooftree}
\AxiomC{$\mf C|_\Gamma \cup \bigl\{\blank = \Mes_p(\Gamma)\bigr\} \vdash \Imp$}
\AxiomC{\hspace{-10pt}$\mf C|_\Gamma \cup \bigl\{\blank = \Mes_q(\Gamma)\bigr\} \vdash \Imp$\hspace{-10pt}}
\AxiomC{$q \comp p$}
\RightLabel{CM}
\TrinaryInfC{$\mf C|_\Gamma \cup \bigl\{\blank = \Mes_{p \vee q}(\Gamma)\bigr\} \vdash \Imp$}
\end{prooftree}
or, more compactly:
\begin{prooftree}
\AxiomC{$\mf C \vdash \Gamma \verifies p$}
\AxiomC{$\mf C \vdash \Gamma \verifies q$}
\AxiomC{$p \comp q$}
\RightLabel{CM}
\TrinaryInfC{$\mf C \vdash \Gamma \verifies p \wedge q$}
\end{prooftree}
\subsection{A Few More Rules}

Let us now present a few more assumptions about the behavior of projective measurements.

\paragraph{Unitary Operator} Considering the application of unitary operators, it is reasonable to assume that if an outcome $p$ is impossible at a slice $\Gamma$ in a circuit~$\mf C$ and if, in $\mf C$, $\Gamma' = U(\Gamma)$ for some unitary o-node $U$ with associated unitary operator $[\![U]\!]$, then outcome $[\![U]\!](p)$ is impossible at $\Gamma'$. By allowing to only apply the unitary operator to a part of $\Gamma$, we obtain the following rule:
\begin{prooftree}
\AxiomC{$\mf C|_{\Gamma :: \Delta} \cup \bigl\{\blank = \Mes_p(\Gamma :: \Delta)\bigr\} \vdash \Imp$}
\AxiomC{$\Gamma' = U(\Gamma) \in \mf C$}
\RightLabel{Uni}
\BinaryInfC{$\mf C|_{\Gamma :: \Delta} \cup \bigl\{\blank = \Mes_{([\![U]\!] \otimes \Id)(p)}(\Gamma' :: \Delta)\bigr\} \vdash \Imp$}
\end{prooftree}
or, more compactly:
\begin{prooftree}
\AxiomC{$\mf C \vdash \Gamma :: \Delta \verifies p$}
\AxiomC{$\Gamma' = U(\Gamma) \in \mf C$}
\RightLabel{Uni}
\BinaryInfC{$\mf C \vdash \Gamma' :: \Delta \verifies ([\![U]\!] \otimes \Id_\Delta)(p)$}
\end{prooftree}
\paragraph{Compatible Preservation} Similarly, we have a commutation rule with measurements yielding compatible outcomes\footnote{In terms of quantum states, this corresponds to the commutation of orthogonal projections on two compatible subspaces.}: if $p$ is impossible at $\Gamma$ and if $\Gamma' = \Mes_q(\Gamma)$ with $q$ compatible with $p$, then $p$ is also impossible at $\Gamma'$. Again, allowing the measurement of $q$ to occur only on a part of $\Gamma$, we~get:

\begin{prooftree}
\AxiomC{$\mf C \vdash \Gamma :: \Delta \verifies p$}
\AxiomC{\hspace{-5pt}$\Gamma' = \Mes_q(\Gamma) \in \mf C$\hspace{-5pt}}
\AxiomC{$p \, \comp \, q \otimes \!\top $}
\RightLabel{CP}
\TrinaryInfC{$\mf C \vdash \Gamma' :: \Delta \verifies p$}
\end{prooftree}


\paragraph{Extending slices}

If a circuit $\mf C$ verifies $p$ at $\Gamma$, this means that cutting the circuit at $\Gamma$ and inserting a measurement of $\Gamma$ with outcome $p^\bot$ leads to an impossible circuit. But in that case, if we cut $\mf C$ along a larger slice $\Gamma :: \Delta$ and insert a measurement of $\Gamma :: \Delta$ with outcome $p^\bot \otimes \! \top$ (the tensor product with $\top$ acting as some form of padding), it is clear that the circuit remains impossible. This leads to the following extension \emph{Tens} (for tensor product) rule:

\begin{prooftree}
\AxiomC{$\mf C \vdash \Gamma \verifies p $}
\RightLabel{Tens}
\UnaryInfC{\vphantom{$\bigl(_($}$\mf C \vdash \Gamma :: \Delta \verifies p \otimes \! \top$}
\end{prooftree}

\paragraph{Permutations of a slice} It is possible, given a slice, to permute the s-nodes it contains. This way, one can obtain a new slice which has exactly the same s-nodes as previously. It should thus be possible to do so, and the next rule will enable this, by indicating how the verification of property is modified by a permutation of the slice.

We will only consider one type of permutation, changing a slice of the form $\Gamma :: \Delta :: \Xi$ to $\Delta :: \Gamma :: \Xi$, which we call a \emph{block swap}. It is easy to verify that block swaps do generate all the possible permutations.
Now, if $\{e_i\}$ (resp. $\{f_j\}$, $\{g_k\}$) is an orthonormal basis of $\mb C^{d(\Gamma)}$ (resp. $\mb C^{d(\Delta)}$, $\mb C^{d(\Xi)}$), then the action of the block swap corresponds to mapping
$ |e_i \otimes f_j \otimes g_k\ket$ to $|f_j \otimes e_i \otimes g_k\ket$.
If we define
$$ \mathrm{swap}_{\Gamma,\Delta,\Xi}(P) \stackrel \Delta = \Bigl\{\sum_{i=1}^{d(\Gamma)} \sum_{j=1}^{d(\Delta)} \sum_{k=1}^{d(\Xi)} \mathopen|f_j \otimes e_i \otimes g_k\ket \bra e_i\otimes f_j \otimes g_k | u \ket \Bigm| u \in P \Bigr\} $$
we then obtain the desired rule for formalizing such permutations:
\begin{prooftree}
\AxiomC{$\mf C \vdash \Gamma :: \Delta :: \Xi \verifies p $}
\RightLabel{Swap}
\UnaryInfC{$\mf C \vdash \Delta :: \Gamma :: \Xi \verifies \mathrm{swap}_{\Gamma, \Delta, \Xi}(p) $}
\end{prooftree}
This ends the first version of our formalism, which is summarized in figure~\ref{fig:first_logic}.

\begin{figure}

\begin{prooftree}
\AxiomC{$\Gamma = \Mes_p(\blank) \in \mf C$}
\RightLabel{NC}
\UnaryInfC{$\mf C \vdash \Gamma \verifies p$}
\end{prooftree}

\begin{prooftree}
\AxiomC{$\vphantom{\bigl\{}$}
\RightLabel{Top}
\UnaryInfC{$\mf C \vdash \Gamma \verifies \top$}
\end{prooftree}

\begin{prooftree}
\AxiomC{$\fall {p \in \mc O} \mf C \vdash \Gamma \verifies p^\bot$}
\RightLabel{Mes}
\UnaryInfC{$\mf C \vdash \Imp$}
\end{prooftree}

\begin{prooftree}
\AxiomC{$\mf C \vdash \Gamma \verifies p$}
\AxiomC{$p \leq q$}
\RightLabel{Ord}
\BinaryInfC{$\mf C \vdash \Gamma \verifies q$}
\end{prooftree}

\begin{prooftree}
\AxiomC{$\mf C \vdash \Gamma \verifies p$}
\AxiomC{$\mf C \vdash \Gamma \verifies q$}
\AxiomC{$p \comp q$}
\RightLabel{CM}
\TrinaryInfC{$\mf C \vdash \Gamma \verifies p \wedge q$}
\end{prooftree}

\begin{prooftree}
\AxiomC{$\mf C \vdash \Gamma \verifies p$}
\RightLabel{Tens}
\UnaryInfC{$\mf C \vdash \Gamma :: \Delta \verifies p \otimes \! \top$}
\end{prooftree}


\begin{prooftree}
\AxiomC{$\mf C \vdash \Gamma :: \Delta \verifies p$}
\AxiomC{$\Gamma' = \Mes_q(\Gamma) \in \mf C $}
\AxiomC{$p \, \comp \, q \otimes \!\top $}
\RightLabel{CP}
\TrinaryInfC{$\mf C \vdash \Gamma' :: \Delta \verifies p$}
\end{prooftree}

\begin{prooftree}
\AxiomC{$\mf C \vdash \Gamma :: \Delta \verifies p$}
\AxiomC{$\Gamma' = U(\Gamma) \in \mf C$}
\RightLabel{Uni}
\BinaryInfC{$\mf C \vdash \Gamma' :: \Delta \verifies ([\![U]\!] \otimes \Id_{\Delta})(p)$}
\end{prooftree}

\begin{prooftree}
\AxiomC{$\mf C \vdash \Gamma :: \Delta :: \Xi \verifies p$}
\RightLabel{Swap}
\UnaryInfC{$\mf C \vdash \Delta :: \Gamma :: \Xi \verifies \mathrm{swap}_{\Gamma, \Delta, \Xi}(p)$}
\end{prooftree}

\caption{Our Logic, first version} \label{fig:first_logic}
\end{figure}

\section{Logical Variations}

We will now present a few results which will simplify this logic some more general rules.

\subsection{The Mes rule, revisited} From the possibilistic non-contextuality rules \emph{Ord} and \emph{Compatible Meet}, it is clear that if~$\{p_1, \ldots, p_n\}$ are mutually compatible elements (so that they all belong to a single boolean subalgebra), then it is equivalent to have
$$ \fall {i \in [\![1, n]\!]} \mf C \vdash \Gamma \verifies p_i 
\qquad \hbox{and} \qquad
\mf C \vdash \Gamma \verifies \bigwedge_i p_i $$
In particular, considering an observable $\mc O = \{ p_1, \ldots, p_n\}$, all the outcomes are mutually compatibles and
$$ \bigwedge_i p_i^\bot = \Bigl(\bigvee_i p_i\Bigr)^\bot = \top^\bot = \bot $$
As a consequence, the \emph{Mes} rule can equivalently be replaced~by
\begin{prooftree}
\AxiomC{$\mf C \vdash \Gamma \verifies \bot\vphantom)$}
\RightLabel{Mes}
\UnaryInfC{$\mf C \vdash \Imp\vphantom)$}
\end{prooftree}


\subsection{The Sasaki Rule}

Suppose now that a graph $\mf C$ is such that $\mf C \vdash \Gamma :: \Delta \verifies p$ and $\Gamma' = \Mes_q(\Gamma)$. We will prove that we have
$$ \mf C \vdash \Gamma' :: \Delta \verifies p \sas (q \otimes \!\top) $$
where the Sasaki operator $p \sas q$ is defined~as
$$ p \sas q \stackrel \Delta {\, = \,} q \wedge \bigl(p \vee q^\bot\bigr) $$
Equivalently, we claim that the following new rule is valid in our logic:
\begin{prooftree}
\AxiomC{$\Gamma :: \Delta \verifies p$}
\AxiomC{$\Gamma' = \Mes_q(\Gamma)$}
\RightLabel{Sas}
\BinaryInfC{$\Gamma' :: \Delta \verifies p \sas (q \otimes \! \top)$}
\end{prooftree}
To show this, we provide a proof in the form of a \emph{proof tree}, where we stack different rules to express chains of reasoning. It reads from top to bottom, where topmost lines correspond to hypotheses, and the bottom line is the conclusion.

\begin{prooftree}
\AxiomC{$\Gamma' = \Mes_q(\Gamma)$}
\RightLabel{NC}
\UnaryInfC{$\Gamma' \verifies q$}
\RightLabel{Tens}
\UnaryInfC{$\Gamma' :: \Delta \vdash q \otimes \! \top$}
\AxiomC{$\Gamma :: \Delta \verifies p$}
\RightLabel{Ord}
\UnaryInfC{$\Gamma :: \Delta \verifies p \vee (q \otimes \! \top)^\bot$}
\AxiomC{$\Gamma' = \Mes_q(\Gamma)$}
\RightLabel{CP}
\BinaryInfC{$\Gamma' :: \Delta \verifies p \vee (q \otimes \! \top)^\bot$}
\RightLabel{CM}
\BinaryInfC{$\Gamma' :: \Delta \verifies \underbrace{(q \otimes \! \top) \wedge \bigl(p \vee (q \otimes \! \top)^\bot\bigr)}_{= \, p \, \sas \, (q \, \otimes \, \! \top)}$}
\end{prooftree}

Let us now show that in the presence of the \emph{Top} and \emph{Ord} rules, the \emph{Sas} rule can replace both the \emph{Non Contradiction} and \emph{Compatible Preservation} rules. We start with the \emph{Non Contradiction} rule, which definition~is
\begin{prooftree}
\AxiomC{$\Gamma = \Mes_p(\blank)$}
\RightLabel{NC}
\UnaryInfC{$\Gamma \verifies p$}
\end{prooftree}
The behaviour of this rule can be obtained using the \emph{Top} and \emph{Sas} rules as follows:
\begin{prooftree}
\AxiomC{$\vphantom{\bigl\{}$}
\RightLabel{Top}
\UnaryInfC{$\Gamma \verifies \top\vphantom($}
\AxiomC{$\Gamma' = \Mes_p(\Gamma)$}
\RightLabel{Sas}
\BinaryInfC{$\Gamma' \verifies p$}
\end{prooftree}
where we have used the fact~that $ \top \sas p = p \wedge (\top \vee p^\bot) = p \wedge \! \top = p $.

Regarding the \emph{Compatible Preservation}, we have to prove that $\Gamma'::\Delta \verifies p$ provided that $\Gamma :: \Delta \verifies p$ and $\Gamma' = \Mes_q(\Gamma)$ with $p$ compatible with $q \otimes \! \top$. This can be achieved the following~way:
\begin{prooftree}
\AxiomC{$\Gamma :: \Delta \verifies p$}
\AxiomC{$\Gamma' = \Mes_q(\Gamma)$}
\RightLabel{Sas}
\BinaryInfC{$\Gamma' :: \Delta \verifies p \sas (q \otimes \! \top)$}
\RightLabel{Ord}
\UnaryInfC{$\Gamma' :: \Delta \verifies p$}
\end{prooftree}
In particular, since $p$ is compatible with $q \otimes \! \top$, we have
$$p \sas (q \otimes \! \top) = p \wedge (q \otimes \! \top) \leq p.$$
As a result, both the \emph{Non Contradiction} and \emph{Compatible Preservation} rules can be replaced by the \emph{Sas} rule we have just introduced.

\subsection{Generalizing the Compatible Meet Rule}

We now show that the compatibility requirement in the \emph{Compatible Meet} rule can be dropped. 
A first step towards this is the following result \cite{Brunet07PLA,Brunet:2009} (we recall that $[\varphi]$ denotes the subspace spanned by a non-zero vector~$|\varphi\ket$):

\begin{proposition} \label{prop:bks}
In a quantum circuit $\mf C$, if a slice $\Gamma$ of dimension at least $3$ is such that $\mf C \vdash \Gamma \verifies [\varphi]$ and $\mf C \vdash \Gamma \verifies [\psi]$ with $[\varphi] \neq [\psi]$, then $\mf C \vdash \Imp$.
\end{proposition}
\begin{proof}[Sketch of Proof]
If $\Gamma \verifies [\varphi]$ and $\Gamma \verifies [\psi]$ with $[\varphi] \neq [\psi]$, then it is possible to construct two finite sequences $([\varphi_k])_{0 \leq k \leq n}$ and $([\psi_k])_{0 \leq k \leq n}$ such that for all $k$ between $0$ and $n$,
$$ \Gamma \verifies [\varphi_k] \cand \Gamma \verifies [\psi_k] $$
and, moreover, $[\varphi_n] \leq [\psi_n]^\bot$. As a consequence, using the \emph{Compatible Meet} rule, we deduce that $\mf C \vdash \Gamma \verifies [\varphi_n] \wedge [\psi_n]$ but $[\varphi_n] \wedge [\psi_n] = \bot $ so that $\mf C \vdash \Imp$.
\end{proof}

Let us now define, given a quantum circuit $\mf C$ and a slice $\Gamma \in \Slice(\mf C)$, the~set
$$ S_{\mf C, \Gamma} = \{p \in L_{d(\Gamma)} \mid \mf C \vdash \Gamma \verifies p\}$$
It is clear that $S_{\mf C, \Gamma}$ is not empty, as it contains $\top$. According to the \emph{Ord} rule, it is closed upwards (that is, if $p \in S_{\mf C, \Gamma}$ and $p \leq q$, then $q \in S_{\mf C, \Gamma}$) and according to \emph{Compatible Meet} it is, indeed, stable by compatible meet. From proposition~\ref{prop:bks}, it cannot contain two distinct atoms unless it contains $\bot$ (in which case the circuit is impossible). We also remark that $S_{\mf C, \Gamma}$ has a finite height (as it is already the case for $L_{d(\Gamma)}$) so that for every element $p \in S_{\mf C, \Gamma}$, there is at least one element $q \in S_{\mf C, \Gamma}$ such that $q \leq p$ and which is minimal in $S_{\mf C, \Gamma}$.

\begin{proposition}[\cite{Brunet:WBR}] 
If $d(\Gamma) \geq 3$, the set $S_{\mf C, \Gamma}$ cannot have two distinct minimal elements.
\end{proposition}
\begin{proof}[Sketch of Proof]
If there were two such minimal elements $p$ and $q$, then we first remark that they cannot be compatible, since otherwise, from the \emph{Compatible Meet} rule, we would have $p \wedge q \in S_{\mf C, \Gamma}$, contradicting their minimality.

Being incompatible, it can be shown that there exists $[\varphi] \leq p$ and $[\psi] \leq q$ such that $[\varphi] \neq [\psi]$ and, putting $c = [\varphi] \vee [\psi]$, such that $p \sas c = [\varphi]$ and $q \sas c = [\psi]$.
Define now $\mf C' = \mf C|_\Gamma \cup \bigl\{\Gamma' = \Mes_c(\Gamma)\bigr\}$. We have $\mf C' \vdash \Gamma' \verifies p \sas c = [\varphi]$ and $\mf C' \vdash \Gamma' \verifies q \sas c = [\psi]$ with $[\varphi] \neq [\psi]$. As a consequence of proposition~\ref{prop:bks}, 
$$ \mf C|_\Gamma \cup \bigl\{\Gamma' = \Mes_c(\Gamma)\bigr\} \vdash \Imp $$
so that $\mf C \vdash \Gamma \verifies c^\bot$ and hence $\mf C \vdash \Gamma \verifies p \sas c^\bot$ using the \emph{Sas} rule. But since $p$ and $c$ are compatible, so are $p$ and $c^\bot$ and $p \sas c^\bot < p$, which contradicts the minimality of $p$ in $S_{\mf C,\Gamma}$.
\end{proof}

\begin{theorem} \label{th:wbr}
If $d(\Gamma) \geq 3$, then there exists an element $k_{\mf C}(\Gamma) \in L_{d(\Gamma)}$ such~that
$$ \fall {p \in L_{d(\Gamma)}} p \in S_{\mf C, \Gamma} \iff k_{\mf C}(\Gamma) \leq p $$
or, equivalently, such~that
$$ S_{\mf C, \Gamma} = k_{\mf C}(\Gamma)^\uparrow = \{ p \in L_{d(\Gamma)} \mid k_{\mf C}(\Gamma) \leq p\}. $$
In the following, $k_{\mf C}(\Gamma)$ will be called the \emph{epistemic state} of $\mf C$ at $\Gamma$.
\end{theorem}
\begin{proof}
Let $e$ be a minimal element of $S_{\mf C, \Gamma}$. For all $p \in S_{\mf C, \Gamma}$, considering a minimal element $f \in S_{\mf C, \Gamma}$ below $p$ (so that $f \leq p$) we have, by unicity of a minimal element, $e = f$ and hence $e \leq p$. This implies that $S_{\mf C, \Gamma} = e^\uparrow$.
\end{proof}

One might worry about the condition $d(\Gamma) \geq 3$. This can, however, be easily circumvented the following way: given a circuit $\mf C$, we consider that it is possible to add an additional s-node $\gamma$ (of dimension at least 3) to~$\mf C$ not connected to any o-node. With this new circuit $\mf C \cup \{\gamma\}$ we can now consider that a verification statement $\mf C \vdash \Gamma \verifies p$ has to be understood~as 
$$ \mf C \cup \{\gamma\} \vdash \Gamma :: [\gamma] \verifies p \otimes \! \top_{\!\{\gamma\}} $$
in which case $d(\Gamma :: [\gamma]) \geq 3$. We will assume that it is always possible to do such a circuit transformation. This way, theorem~\ref{th:wbr} always applies and the \emph{Compatible Meet} rule can be replaced by the more general \emph{Meet} rule:

\begin{prooftree}
\AxiomC{$\mf C \vdash \Gamma \verifies p$}
\AxiomC{$\mf C \vdash \Gamma \verifies q$}
\RightLabel{Meet}
\BinaryInfC{$\mf C \vdash \Gamma \verifies p \wedge q$}
\end{prooftree}




This ends our discussion leading to the final version of our logic, which is presented in figure~\ref{fig:last_logic}. However, one can remark that the \emph{Mes}. This disparition will be discussed in section~\ref{sub:meaning}.

\begin{figure}
\begin{prooftree}
\AxiomC{\vphantom{$($}}
\RightLabel{Top}
\UnaryInfC{$\Gamma \verifies \top \vphantom)$}
\end{prooftree}


\begin{prooftree}
\AxiomC{$\Gamma \verifies p \vphantom)$}
\AxiomC{$p \leq q$}
\RightLabel{Ord}
\BinaryInfC{$\Gamma \verifies q \vphantom)$}
\end{prooftree}

\begin{prooftree}
\AxiomC{$\Gamma \verifies p$}
\AxiomC{$\Gamma \verifies q$}
\RightLabel{Meet}
\BinaryInfC{$\Gamma \verifies p \wedge q \vphantom)$}
\end{prooftree}

\begin{prooftree}
\AxiomC{$\Gamma :: \Delta :: \Xi \verifies p$}
\RightLabel{Swap}
\UnaryInfC{$\Delta :: \Gamma :: \Xi \verifies \mathrm{swap}_{\Gamma, \Delta, \Xi}(p)$}
\end{prooftree}

\begin{prooftree}
\AxiomC{$\Gamma :: \Delta \verifies p$}
\AxiomC{$\Gamma' = \Mes_q(\Gamma)$}
\RightLabel{Sas}
\BinaryInfC{$\Gamma' :: \Delta \verifies p \sas (q \otimes \! \top)$}
\end{prooftree}

\begin{prooftree}
\AxiomC{$\Gamma :: \Delta \verifies p$}
\AxiomC{$\Gamma' = U(\Gamma)$}
\RightLabel{Uni}
\BinaryInfC{$\Gamma' :: \Delta \verifies (\bracc U \otimes \mathrm{Id}_{\Delta})(p)$}
\end{prooftree}

\begin{prooftree}
\AxiomC{$\Gamma \verifies p$}
\RightLabel{Tens}
\UnaryInfC{$\Gamma :: \Delta \verifies p \otimes \! \top$}
\end{prooftree}


\caption{Our logic, final version} \label{fig:last_logic}
\end{figure}

\subsection{Quantum Teleportation}

In order to illustrate the expressivity of our formalism, let us apply it to the circuit presented in figure~\ref{fig:T_uncut}. It can be seen as a teleportation scheme \cite{Bennett93Teleportation} with the creation of a Bell pair (at $(B_3, C_2)$) and then teleporting $A_1$ at $C_3$.

It can be remarked that the choice of the operator applied between $C_2$ and $C_3$ is determined by the outcome of the measurements between $A_3$ and $A_4$, and between $B_4$ and $B_5$. This circuit thus only represents one of the four possibilities.

Let us first focus on the preparation of the Bell pair. We have:
\begin{prooftree}
\AxiomC{\vphantom(}
\RightLabel{Top}
\UnaryInfC{$B_0 \verifies \top$}
\AxiomC{$B_1 = \Mes_{[1]}(B_0)$}
\RightLabel{Sas}
\BinaryInfC{$B_1 \verifies [1]$}
\RightLabel{Tens}
\UnaryInfC{$B_1, C_1 \verifies [1] \otimes \top$}
\AxiomC{\vphantom(}
\RightLabel{Top}
\UnaryInfC{$C_0 \verifies \top$}
\AxiomC{$C_1 = \Mes_{[1]}(C_0)$}
\RightLabel{Sas}
\BinaryInfC{$C_1 \verifies [1]$}
\RightLabel{Tens}
\UnaryInfC{$C_1, B_1 \verifies [1] \otimes \top$}
\RightLabel{Swap}
\UnaryInfC{$B_1, C_1 \verifies \top \otimes [1]$}
\RightLabel{Meet}
\BinaryInfC{$B_1, C_1 \verifies [1] \otimes [1]$}
\end{prooftree}
with $[1] \otimes [1] = [11]$. From this, we deduce
\begin{prooftree}
\AxiomC{$B_1, C_1 \verifies [11]$}
\AxiomC{$B_2 = H(B_1)$}
\RightLabel{Uni}
\BinaryInfC{$B_2, C_1 \verifies [|01\ket - |11\ket]$}
\end{prooftree}
since $[|01\ket - |11\ket] = ([\![H]\!] \otimes \Id) [11]$,~and
\begin{prooftree}
\AxiomC{$B_2, C_1 \verifies [|01\ket - |11\ket]$}
\AxiomC{$B_3, C_2 = \mathrm{CNot}(B_2, C_1)$}
\RightLabel{Uni}
\BinaryInfC{$B_3, C_2 \verifies [|01\ket - |10\ket]$}
\end{prooftree}
We thus have shown that $ B_3, C_2 \verifies [|01\ket - |10\ket] $. 

Let us now move to the second part of the circuit, and suppose that particle~$A_1$ is possibly entangled with another quantum system denoted~$R_1$ in such a way that they jointly verify some property~$p$:
$$ A_1, R_1 \verifies p $$
Let us write $p = \Vect\{|0 \, a_i\ket + |1 \, b_i\ket\}_{i \in \mc I}$. First, combining $B_3, C_2 \verifies [|01\ket - |10\ket]$ with $A_1, R_1 \verifies \Vect\{|0 \, a_i\ket + |1 \, b_i\ket\}_{i \in \mc I}$ using adequate \emph{Tens}, \emph{Meet} and \emph{Swap} rules, we obtain:
$$ B_3, C_2, A_1, R_1 \verifies \Vect\bigl\{|0 1 0 \, a_i\ket + |0 1 1 \, b_i\ket - |1 0 0 \, a_i\ket - |1 0 1 \, b_i\ket\bigr\} $$
Since $A_2, B_4 = \mathrm{CNot}(A_1, B_3)$, we deduce
$$ B_4, C_2, A_2, R_1 \verifies \Vect\bigl\{|0 1 0 \, a_i\ket + |1 1 1 \, b_i\ket - |1 0 0 \, a_i\ket - |0 0 1 \, b_i\ket\bigr\} $$
Applying an Hadamard gate from $A_2$ to $A_3$ leads~to
\begin{multline*}
B_4, C_2, A_3, R_1 \verifies \Vect \bigl\{|0 1 0 \, a_i\ket + |0 1 1 \, a_i\ket + |1 1 0 \, b_i\ket \\ - |1 1 1 \, b_i\ket - |1 0 0 \, a_i\ket - |1 0 1 \, a_i\ket - |0 0 0 \, b_i\ket + |0 0 1 \, b_i\ket \bigr\}
\end{multline*}
Then, measuring $A_3$ with outcome $[1]$~implies
$$ B_4, C_2, A_4, R_1 \verifies \Vect\bigl\{|0 1 1 \, a_i\ket - |1 1 1 \, b_i\ket - |1 0 1 \, a_i\ket + |0 0 1 \, b_i\ket \bigr\} $$
and finally, measuring $B_4$ with outcome $[0]$,
$$ B_5, C_2, A_4, R_1 \verifies \Vect\bigl\{|0 1 1 \, a_i\ket + |0 0 1 \, b_i\ket \bigr\} $$
so that
$$ B_5, A_4, C_2, R_1 \verifies \Vect\bigl\{|0 1 1 \, a_i\ket + |0 1 0 \, b_i\ket \bigr\} $$
But $|0 1 1 \, a_i\ket + |0 1 0 \, b_i\ket = \mathopen|01\ket \otimes \bigl(\mathopen|1 \, a_i\ket + |0 \, b_i\ket\bigr)$, so that
\begin{prooftree}
\AxiomC{$B_5, A_4, C_2, R_1 \verifies [01] \otimes \Vect\bigl\{|1 \, a_i\ket + |0 \, b_i\ket \bigr\}$}
\RightLabel{Ord}
\UnaryInfC{$B_5, A_4, C_2, R_1 \verifies \top \! \otimes \Vect\bigl\{|1 \, a_i\ket + |0 \, b_i\ket \bigr\}$}
\end{prooftree}
Just apply a Not-gate to $C_2$ and we obtain:
$$ B_5, A_4, C_3, R_1 \verifies \top \! \otimes \Vect\bigl\{|0 \, a_i\ket + |1 \, b_i\ket \bigr\}, $$
that is $B_5, A_4, C_3, R_1 \verifies \top \otimes p$.

Thus we have shown that for all $p$, from $A_1, R_1 \verifies p$, we deduce
$$B_5, A_4, C_3, R_1 \verifies \top \! \otimes p.$$
Obviously, if different outcomes were found at $A_3$ and $B_4$, applying the corresponding operator at $C_2$ would lead to the same statement. This illustrates that after the application of the circuit, any property regarding the possibility of measurement outcomes verifies by the system $A_1, R_1$ has be ``transfered'' to~$C_3, R_1$. This is the rigourous expression, in terms of verification statements, of the fact that~$A$ \emph{seems} to have been teleported to~$C$: any property previously verifies by~$A$ is now verified by~$C$.

In this analysis, we stress again the fact that all our statements is of epistemic nature: verification statemens only deal with the possibility or impossibility of obtaining specific measurement outcomes.

\section{Some More Properties}


\subsection{Verification and Full Subgraphes}

First, let us study how the provability of verification statements is preserved when one considers full subgraphes of a circuit.

\begin{proposition}[Monotony] \label{prop:monotony}
Given two circuits $\mf C$ and $\mf C'$ such that $\mf C \contains \mf C'$, if $\mf C \vdash \Gamma \verifies p$, then $\mf C' \vdash \Gamma \verifies p$.
\end{proposition}
\begin{proof}
This is a direct consequence of the fact that if $\mf C \contains \mf C'$, then one moreover has $\Slice(\mf C) \subseteq \Slice(\mf C')$, so that any proof of $\Gamma \vdash p$ in $\mf C$ is also valid in $\mf C'$.
\end{proof}

This result can be expressed in terms of epistemic states the following way:
$$ \mf C \contains \mf C' \implies \fall {\Gamma \in \Slice(\mf C)} k_{\mf C'}(\Gamma) \leq k_{\mf C}(\Gamma) $$

\begin{corollary}
If $\mf C'$ is possible and if $\mf C$ is a full subgraph of $\mf C'$, then $\mf C$ is also possible.
\end{corollary}

Let now determine a full subgraph of a circuit $\mf C$ which is sufficient for proving a statement of the form $\Gamma \verifies p$. In order to do this, let us introduce the notion of \emph{strong past}.

\begin{definition}[Strong Past]
Given a slice $\Gamma$ of a circuit $\mf C$, an s-node $n$ of $\mf C$ is in the \emph{strong past} of $\Gamma$ is every path going out from $n$ crosses $\Gamma$.
\end{definition}


\begin{proposition} \label{prop:sp_is_full}
For all $\Gamma \in \Slice(\mf C)$, we have $ \mathrm{sp}(\mf C, \Gamma) \contains \mf C$.
\end{proposition}

\begin{proof}
Let $a, b$ be two s-nodes in the strong past of $\Gamma$ and suppose that $a \transition {\mf C} b$. Let $n$ be a node in the path from $a$ et $b$. Any path going out from $n$ can be completed into a path going out from $a$. But since $a$ is in the strong past of $\Gamma$, this path intersects $\Gamma$. As a consequence, $n$ is also in the strong past of $\Gamma$.
\end{proof}

\begin{proposition} \label{prop:sp_is_enough}
If $\mf C \vdash \Gamma \verifies p$, then $\mathrm{sp}(\mf C, \Gamma) \vdash \Gamma \verifies p$.
\end{proposition}
\begin{proof}
This follows from the fact that in all the rules in figure~\ref{fig:last_logic}, the slices of the premisses are in the strong past of the slice of the conclusion. As a consequence, any proof of $\Gamma \verifies p$ in $\mf C$ is also valid in $\mathrm{sp}(\mf C, \Gamma)$.
\end{proof}

As a consequence of these results, any verification statement about a slice $\Gamma$ in $\mf C$ can proven by only considering the strong past of $\Gamma$ in $\mf C$:

\begin{theorem}[Strong Causality] \label{th:strong_causality}
For all $\Gamma \in \Slice(\mf C)$ and $p \in L_{d(\Gamma)}$,
$$ \mf C \vdash \Gamma \verifies p \iff \mathrm{sp}(\mf C, \Gamma) \vdash \Gamma \verifies p $$
\end{theorem}
\begin{proof}
This is a direct consequence of propositions~\ref{prop:monotony}, \ref{prop:sp_is_full} and~\ref{prop:sp_is_enough}.
\end{proof}


\subsection{The Meaning of Verification} \label{sub:meaning}

Initially, the verification statement $\mf C \vdash \Gamma \verifies p$ was defined as the statement that appending a measurement $\blank = \Mes_{p^\bot}(\Gamma)$ to $\mf C|_\Gamma$ would lead to an impossible circuit. However, in the definition of our logic, it appeared that in addition to the definition of verification statements, the only place where $\Imp$ was present was in the \emph{Mes} rule, as the consequence of the verification of $\bot$ at some slice of a circuit.

However, in our logic as definied in figure~\ref{fig:last_logic}, there is need any longer to references to $\Imp$, and $\verifies$ can be considered as an atomic statement rather some syntactical sugar as it was previously the case. Indeed, the next results show that our logic correctly captures the intended meaning of our verification statement as we will prove~that
$$ \mf C \vdash \Gamma \verifies p \iff \mathrm{sp}(\mf C, \Gamma) \cup \{\Gamma' = \Mes_{p^\bot}(\Gamma)\} \vdash \Gamma' \verifies \bot $$

\begin{proposition} \label{prop:compatible_slices}
Suppose now that $\Gamma' = \Mes_p(\Gamma)$ in a circuit~$\mf C$, and that~$A$ is a slice of~$\mf C$ compatible with $\Gamma'$, by which we mean that there exists a slice $\Xi \in \Slice(\mf C)$ such that
$$ \Gamma'_{\{\}} \cup A_{\{\}} \subseteq \Xi_{\{\}}$$
One can then define $\Delta$ and $B$ such that both~$\Gamma' :: \Delta$ and~$A :: B$ are slices of~$\mf C$, and that 
$$ (\Gamma' :: \Delta)_{\{\}} = (A :: B)_{\{\}} = \Gamma'_{\{\}} \cup A_{\{\}} $$
If $\sigma$ denotes the unitary operator obtained from successive applications of the \emph{Swap} rule for going from $A :: B$ to $\Gamma' :: \Delta$, then for all $q \in L_{d(A)}$:
$$ \mf C \vdash A \verifies q \implies k_{\mf C}(\Gamma :: \Delta) \sas (p \otimes \! \top) \leq \sigma(q \otimes \! \top) $$
\end{proposition}
\begin{proof} We prove this by induction on the proof tree leading to $A \verifies q$.
The proof for the \emph{Top}, \emph{Ord}, \emph{Meet}, \emph{Tens} and \emph{Swap} rules is direct. Suppose now that the root rule is an instance of the \emph{Sas} rule, of the~form
\begin{prooftree}
\AxiomC{$U::V \verifies a$}
\AxiomC{$U' = \Mes_b(U)$}
\RightLabel{Sas}
\BinaryInfC{$U' :: V \verifies a \sas (b \otimes \! \top_{\!V})$}
\end{prooftree} 
with $A = U' :: V$.

Suppose first that $U'_{\{\}} \cap \Gamma'_{\{\}} = \emptyset$. This implies that $U::V$ is also compatible with $\Gamma'$. Let us define $\Delta_1$ and $\Delta_2$ such that
$$ V :: B = \Delta_1 :: \Gamma' :: \Delta_2 $$
By induction hypothesis, one has
$$ k_{\mf C}(U :: \Delta_1 :: \Gamma :: \Delta_2) \sas (\top_{\!U :: \Delta_1} \otimes p \otimes \! \top_{\!\Delta_2}) \leq a \otimes \! \top_{\!B} $$
But
$$ k_{\mf C}(U' :: \Delta_1 :: \Gamma :: \Delta_2) \leq k_{\mf C}(U :: \Delta_1 :: \Gamma :: \Delta_2) \sas (b \otimes \! \top_{\!V :: B}) $$
and for all $q$,
$$ \bigl(q \sas (b \otimes \! \top_{\!V :: B})\bigr) \sas (\top_{\!U :: \Delta_1} \otimes p \otimes \! \top_{\!\Delta_2}) = \bigl(q \sas (\top_{\!U :: \Delta_1} \otimes p \otimes \! \top_{\!\Delta_2})\bigr) \sas (b \otimes \! \top_{\!V :: B}) $$
so that
\begin{multline*}
k_{\mf C}(U' :: \Delta_1 :: \Gamma :: \Delta_2) \sas (\top_{\!U :: \Delta_1} \otimes p \otimes \! \top_{\!\Delta_2}) \\
\leq \bigl(k_{\mf C}(U :: \Delta_1 :: \Gamma :: \Delta_2) \sas (b \otimes \! \top_{\!V :: B})\bigr) \sas (\top_{\!U :: \Delta_1} \otimes p \otimes \! \top_{\!\Delta_2}) \\
\leq \bigl(k_{\mf C}(U :: \Delta_1 :: \Gamma :: \Delta_2) \sas (\top_{\!U :: \Delta_1} \otimes p \otimes \! \top_{\!\Delta_2})\bigr) \sas (b \otimes \! \top_{\!V :: B}) \\
\leq (a \otimes \! \top_{\!B}) \sas (b \otimes \! \top_{\!V :: B}) \leq \bigl(a \sas (b \otimes \! \top_{\!V})\bigr) \otimes \! \top_{\! B}
\end{multline*}
which is the expected result.

Otherwise, $U'_{\{\}} \cap \Gamma'_{\{\}} \neq \emptyset$ and the only possibility is to have the application of the \emph{Sas} rule with $U = \Gamma$, $U' = \Gamma'$, $b = p$ and $q = k_{\mf C}(\Gamma :: V) \sas (p \otimes \! \top_{\!V})$, in which case the result follows directly.

The treatment of the \emph{Uni} rule is similar to that of the \emph{Sas} rule.
\end{proof}

\begin{proposition}
If $\Gamma :: \Delta \in \Slice(\mf C)$ and $\Gamma' = \Mes_p(\Gamma) \in \mf C$,~then
$$ k_{\mf C}(\Gamma' :: \Delta) = k_{\mf C}(\Gamma :: \Delta) \sas (p \otimes \! \top) $$
\end{proposition}
\begin{proof}
First, let us remark that
\begin{prooftree}
\AxiomC{$\Gamma :: \Delta \verifies k_{\mf C}(\Gamma :: \Delta)$}
\AxiomC{$\Gamma' = \Mes_p(\Gamma)$}
\BinaryInfC{$\Gamma' :: \Delta \verifies k_{\mf C}(\Gamma :: \Delta) \sas (p \otimes \! \top)$}
\end{prooftree}
so that
$$k_{\mf C}(\Gamma' :: \Delta) \leq k_{\mf C}(\Gamma :: \Delta) \sas (p \otimes \! \top).$$
Conversely, considering proposition~\ref{prop:compatible_slices} with $A = \Gamma' :: \Delta$ (and hence $B$ is the empty slice) and $q = k_{\mf C}(\Gamma' :: \Delta)$, we have
$$ k_{\mf C}(\Gamma :: \Delta) \sas (p \otimes \! \top) \leq k_{\mf C}(\Gamma' :: \Delta) $$
\end{proof}

The previous result now allows us to recover and refine the initial meaning of our verification statements:

\begin{theorem} \label{th:verification_is_atomic}
For all $\Gamma \in \Slice(\mf C)$ and $p \in L_{d(\Gamma)}$,
$$ \mf C \vdash \Gamma \verifies p \iff \mathrm{sp}(\mf C, \Gamma) \cup \{\Gamma' = \Mes_{p^\bot}(\Gamma)\} \vdash \Gamma' \verifies \bot $$
\end{theorem}
\begin{proof}
Obviously, if $\mf C \vdash \Gamma \verifies p$, then by putting
$$ \mf C' = \mathrm{sp}(\mf C, \Gamma) \cup \{\Gamma' = \Mes_{p^\bot}(\Gamma)\}, $$
we have
\begin{prooftree}
\AxiomC{$\mf C' \vdash \Gamma \verifies p$}
\AxiomC{$\mf C' \vdash \Gamma' = \Mes_{p^\bot}(\Gamma)$}
\RightLabel{Sas}
\BinaryInfC{$\mf C' \vdash \Gamma' \verifies \underbrace{p \sas p^\bot}_{\, \bot}$}
\end{prooftree}
Conversely, if $\mf C' \vdash \Gamma' \verifies \bot$, then
$$ \bot = k_{\mf C'}(\Gamma') = k_{\mf C'}(\Gamma) \sas p^\bot = k_{\mathrm{sp}(\mf C, \Gamma)}(\Gamma) \sas p^\bot$$
But in orthomodular lattice, it is true that
$$ a \sas b \iff a \leq b^\bot $$
so that $k_{\mathrm{sp}(\mf C, \Gamma)}(\Gamma) \leq p$ and hence $\mathrm{sp}(\mf C, \Gamma) \vdash \Gamma \verifies p$ and, finally,
$ \mf C \vdash \Gamma \verifies p$.
\end{proof}

\subsection{Knowledge and Entanglement}

Given a slice $\Gamma :: \Delta$ of a circuit $\mf C$, for all $p$ and $q$ in $L_{d(\Gamma)}$, if $\Gamma :: \Delta$ verifies both $p \otimes \! \top_{\!\Delta}$ and $q \otimes \! \top_{\!\Delta}$, it also verifies their meet $(p \wedge q) \otimes \! \top_{\!\Delta}$. Moreover $\Gamma :: \Delta \verifies \top_{\!\Gamma} \otimes \! \top_{\!\Delta}$. 

Let $\Gamma \verifies p \atfter \Delta$ denote $\Gamma :: \Delta \verifies p \otimes \top_{\!\Delta}$ which reads ``$\Gamma$ verifies $p$ at $\Delta$''. We thus have:
\begin{gather*}
\Gamma \verifies \top \atfter \Delta \\
\bigl(\Gamma \verifies p \atfter \Delta\ \hbox{and}\ p \leq q \bigr) \implies \Gamma \verifies q \atfter \Delta \\ 
\bigl(\Gamma \verifies p \atfter \Delta \ \hbox{and}\ \Gamma \verifies q \atfter \Delta\bigr) \implies \Gamma \verifies p \wedge q \atfter \Delta
\end{gather*}
This suggest the following definition:
\begin{definition}
For all $\Gamma :: \Delta \in \Slice(\mf C)$, 
$$ k_{\mf C}(\Gamma \mid \Delta) = \min \bigl\{p \in L_\Gamma \bigm| \Gamma \verifies p \atfter \Delta \bigr\} $$
\end{definition}

Obviously, because of the \emph{Tens} rule, one has $k_{\mf C}(\Gamma \mid \Delta) \leq k_{\mf C}(\Gamma)$. This result can actually be significantly strenghened as follows:
\begin{proposition}
If both $\Gamma :: \Delta$ and $\Gamma :: \Delta'$ are slices of $\mf C$ and if $\Delta$ is in the strong past of $\Delta'$, then for all $p \in L_{d(\Gamma)}$,
$$ \Gamma \verifies p \atfter \Delta \implies \Gamma \verifies p \atfter \Delta' $$
\end{proposition}
\begin{proof}
This can be proved in a way similar to the proof of theorem~\ref{th:verification_is_atomic}.
\end{proof}

This suggest that ``$\atfter$'' should be pronounced ``\emph{after (inclusive)}'' rather than just~``\emph{at}''. A direct consequence of this is the following:
\begin{theorem} \label{theo:at_reads_after}
If both $\Gamma :: \Delta$ and $\Gamma :: \Delta'$ are slices of $\mf C$ and if $\Delta$ is in the strong past of $\Delta'$, then
$$ k_{\mf C}(\Gamma \mid \Delta') \leq k_{\mf C}(\Gamma \mid \Delta) $$
\end{theorem}

This shows that if the system $S_1$ at $\Gamma$ is entangled with the system $S_2$ at~$\Delta$, then if $S_1$ is left untouched, acting on $S_2$ (and, in particular, performing measurements on $S_2$) can only increasing one's knowledge about $S_1$.

\ 

To illustrate this, let us consider again the example, taken from \cite{Aharonov-Albert:IsUsualII}, of a particle which can be found in three boxes $A$, $B$ and $C$. after having been prepared in a state $|001\ket + |010\ket + |001\ket$ (where we indicate the modes in the different boxes). The next circuit illustrate the situation where the particle is found in box~$B$:


\hbox{} \hfill %
\begin{tikzpicture}[point/.style={circle,inner sep=0pt,outer sep = 0pt, minimum size=2pt,fill=black}]
	\matrix[row sep=4mm,column sep=5mm]{
		\node (a1) [point] {}; & \node[rectangle] (A) {$0$} ; & \node (a2) [point] {}; \\
		\node (b1) [point] {}; & \node[rectangle] (B) {$1$} ; & \node (b2) [point] {}; \\
		\node (c1) [point] {}; & \node[rectangle] (C) {$0$} ; & \node (c2) [point] {}; \\
	};
\node [above = -2pt of a1] {$A_1$} ;
\node [above = -2pt of b1] {$B_1$} ;
\node [above = -2pt of c1] {$C_1$} ;
\node [above = -2pt of a2] {$A_2$} ;
\node [above = -2pt of b2] {$B_2$} ;
\node [above = -2pt of c2] {$C_2$} ;
\draw (a1) -- (A) ; \draw [very thin] (A) -- (a2) ;
\draw (b1) -- (B) -- (b2) ;
\draw (c1) -- (C) ; \draw [very thin] (C) -- (c2) ;
\draw [thick] (A.north west) -- (A.south west) ;
\draw [thick] (A.north east) -- (A.south east) ;
\draw [thick] (B.north west) -- (B.south west) ;
\draw [thick] (B.north east) -- (B.south east) ;
\draw [thick] (C.north west) -- (C.south west) ;
\draw [thick] (C.north east) -- (C.south east) ;
\draw[decorate,decoration={brace}, thick] ($(c1) + (-9pt, -4pt)$) -- ($(a1) + (-9pt, 4pt)$)
  node[midway, left] (bracket) {$\bigl[\mathopen|A\ket + \mathopen|B\ket + \mathopen|C\ket \bigr]\ $} ;
\end{tikzpicture} %
\hfill \hbox{}

Let us first compute $k(A_1, B_1 \mid C_1)$. We have:
$$ [|100\ket+|010\ket+|001\ket] \leq P \otimes \top \iff P^\bot \otimes \top \leq [|100\ket+|010\ket+|001\ket]^\bot $$
so that
\begin{multline*}
[\varphi] \in P^\bot \iff \Bigl([\varphi]\otimes[0] \in [|100\ket+|010\ket+|001\ket]^\bot \cand \\
[\varphi]\otimes[1] \in [|100\ket+|010\ket+|001\ket]^\bot\Bigr)
\end{multline*}
Putting $|\varphi\ket = a \mathopen|00\ket + b \mathopen|01\ket + c \mathopen|10\ket + d \mathopen|11\ket $, this implies
$$ b + c = 0 \cand a = 0 $$
so that $P^\bot = [|01\ket - |10\ket] + [11]$, $ P = [00] + [|01\ket + |10\ket] $ and finally
$$ k_{\mf C}(A_1, B_1 \mid C_1) = [00] + [|01\ket + |10\ket]. $$

Later, we have
$$ k_{\mf C}(A_1, B_1, C_2) = [|100\ket+|010\ket+|001\ket] \sas (\top \! \otimes [0]) = [|100\ket+|010\ket] $$
so that
$$ k_{\mf C}(A_1, B_1 \mid C_2) = [|01\ket + |10\ket] $$
We thus have found that
$$ k_{\mf C}(A_1, B_1 \mid C_2) = [|01\ket + |10\ket] \leq [00] + [|01\ket + |10\ket] = k_{\mf C}(A_1, B_1 \mid C_1) $$
which illustrates the fact that the knowledge regarding $[A_1, B_1]$ increases when the opening of box $C$ teaches us that the particle is not there.

\section{Realism and Lorentz Invariance}

Let us now turn to the question whether it is possible to have a realistic and Lorentz-invariant interpretation of quantum mechanics. Following Einstein, Podolsky and Rosen \cite{Einstein35EPR}, and using subsequent amendments by Redhead \cite{Redhead:Realism},
\begin{quotation}
``If we can predict with certainty (or at any rate with probability one) the result of measuring a physical quantity at time $t$, then at the time $t$, there exists \emph{an element of reality} corresponding to this physical quantity and having a value equal to the predicted measurement result.'' 
\end{quotation} 
and we will consider the following definition of Lorentz invariance, borrowed from \cite{Vaidman93:EltsReality}:
\begin{quotation}
``If an element of reality corresponding to some Lorentz-invariant physical quantity exists and has a value within space-time region $R$ with respect to one space-like hypersurface containing $R$, then it exists and has the same value in $R$ with respect to any other hypersurface containing $R$.''
\end{quotation}
In \cite{Hardy92}, Hardy presents a gedanken experiment which, he argues, shows that it is not possible to have a realistic Lorentz-invariant quantum theory. In the same period, a similar argument was proposed by Clifton, Pagonis and Pitowsky \cite{CliftonPagonisPitowsky92} using three particles prepared in a GHZ-like state \cite{GHZ90}.

We will argue, on the contrary, that it is possible to have a realistic Lorentz-invariant interpretation of quantum mechanics, with elements of reality corresponding to verification statements, i.e.\ statements of the~form
$$ \Gamma \verifies p $$

In order to illustrate this, let us first describe Hardy's gedanken experiment in our formalism. The setup consists in two Mach-Zender-type interferometers, one for positrons and one for electrons. The key point is that the two interferometers have overlapping arms, so that if a positron and an electron both take these overlapping arms, they annihilate each other. The corresponding circuit is represented in figure~\ref{fig:hardy}, where the $A$ area represents the overlapping zone.

The action of the different beamsplitters is given by the following mappings:
\begin{align*}
|e^\pm\ket & \mapsto \frac 1 {\sqrt 2} \bigl(|v^\pm\ket + i |w^\pm\ket\bigr) \\
|v^\pm\ket & \mapsto \frac 1 {\sqrt 2} \bigl( i |c^\pm\ket + |d^\pm\ket \bigr) \\
|u^\pm\ket & \mapsto \frac 1 {\sqrt 2} \bigl( |c^\pm\ket + i |d^\pm\ket \bigr)
\end{align*}

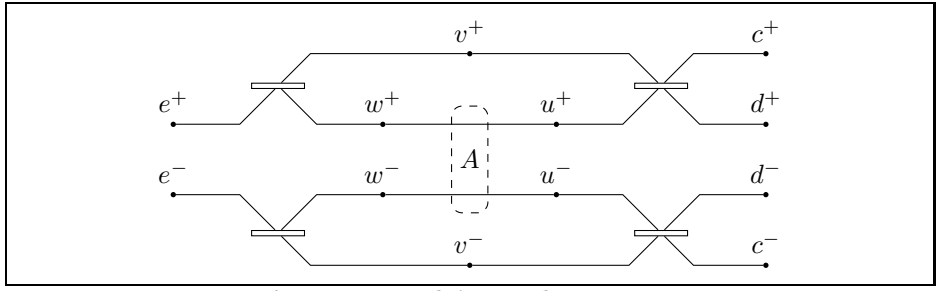
\begin{figure}
\begin{centering}
\begin{tikzpicture}[point/.style={circle,inner sep=0pt,outer sep = 0pt, minimum size=2pt,fill=black},
bsp/.style={rectangle, draw, minimum width=20pt, minimum height = 2pt, inner sep = 0pt}]
	\matrix[row sep=10pt,column sep=10mm]{
		& & & \node (vp) [point] {} ; & & & \node (cp) [point] {} ; \\
		& \node(bs1p) [bsp] {} ; & & & & \node(bs2p) [bsp] {} ; \\
		\node (ep) [point] {} ; & & \node (wp) [point] {} ; & \node (ap) {} ; & \node (up) [point] {} ; & & \node (dp) [point] {} ; \\
		\\
		\node (em) [point] {} ; & & \node (wm) [point] {} ; & \node (am) {} ; & \node (um) [point] {} ; & & \node (dm) [point] {} ; \\
		& \node (bs1m) [bsp] {} ; & & & & \node(bs2m) [bsp] {} ; \\
		& & & \node (vm) [point] {} ; & & & \node (cm) [point] {} ; \\
	};
	\draw
		let \p1 = (vp), \p2 = (bs1p), \n1 = {\y1 - \y2} in (bs1p) -- ($(bs1p) + (\n1, \n1)$) -- (vp) -- ($(bs2p) + (-\n1, \n1)$) -- (bs2p) -- ($(bs2p) + (\n1, \n1)$) -- (cp) ;
	\draw 
		let \p1 = (ep), \p2 = (bs1p), \n1 = {\y2 - \y1} in (ep) -- ($(bs1p) + ({-\n1}, {-\n1})$) -- (bs1p) -- ($(bs1p) + (\n1, {-\n1})$) -- (wp) -- (up) -- ($(bs2p)+ (-\n1, -\n1)$) -- (bs2p) -- ($(bs2p)+ (\n1, -\n1)$) -- (dp) ;
	\draw 
		let \p1 = (em), \p2 = (bs1m), \n1 = {\y1 - \y2} in (em) -- ($(bs1m) + (-\n1, \n1)$) -- (bs1m) -- ($(bs1m) + (\n1, \n1)$) -- (wm) -- (um) -- ($(bs2m)+ (-\n1, \n1)$) -- (bs2m) -- ($(bs2m)+ (\n1, \n1)$) -- (dm) ;
	\draw
		let \p1 = (vm), \p2 = (bs1m), \n1 = {\y2 - \y1} in (bs1m) -- ($(bs1m) + (\n1, -\n1)$) -- (vm) -- ($(bs2m) + (-\n1, -\n1)$) -- (bs2m) -- ($(bs2m) + (\n1, -\n1)$) -- (cm) ;
	\node[dashed, draw, fit=(ap) (am),rounded corners] (fit) {} ;
	\node at (fit) {$A$} ;
	\node (epl) [above = 0pt of ep] {$e^+$} ;
	\node (vpl) [above = 0pt of vp] {$v^+$} ;
	\node (wpl) [above = 0pt of wp] {$w^+$} ;
	\node (upl) [above = 0pt of up] {$u^+$} ;
	\node (vml) [above = 0pt of vm] {$v^-$} ;
	\node (wml) [above = 0pt of wm] {$w^-$} ;
	\node (uml) [above = 0pt of um] {$u^-$} ;
	\node (eml) [above = 0pt of em] {$e^-$} ;
	\node (cpl) [above = 0pt of cp] {$c^+$} ;
	\node (dpl) [above = 0pt of dp] {$d^+$} ;
	\node (cml) [above = 0pt of cm] {$c^-$} ;
	\node (dml) [above = 0pt of dm] {$d^-$} ;
\end{tikzpicture} \\
\end{centering}
\caption{Hardy's paradox circuit $\mf H$} \label{fig:hardy}
\end{figure}

Let's formalize the behaviour of this circuit considering particle modes. Past the first beamsplitters, one~has
$$ v^+,w^+ \verifies \bigl[|10\ket + i|01\ket\bigr] \qquad \hbox{and} \qquad v^-, w^- \verifies \bigl[|10\ket + i|01\ket\bigr], $$
the combination of which yielding
$$ v^+, w^+, w^-, v^- \verifies \bigl[|1001\ket + i |1010\ket + i |0101\ket - |0110\ket\bigr]. $$
In the \emph{annihilation} zone $A$, the term $|0110\ket$ -- corresponding to having both the particle and antiparticle take the intersecting arms and thus annihilating each other -- becomes $|0000\ket$, leading~to
$$ v^+, u^+,u^-,v^- \verifies \bigl[|1001\ket + i |1010\ket + i |0101\ket - |0000\ket\bigr] $$

Considering now the slice $(c^+, d^+, u^-, v^-)$ where the positron has past the second beamsplitter while the electron has not, we obtain:
$$ c^+, d^+, u^-, v^- \verifies\bigl[2 i|1001\ket - |1010\ket + i |0110\ket - \sqrt 2 |0000\ket \bigr] $$
Finally, one both particles have past their second beamsplitter, one~has
$$ c^+, d^+, d^-, c^- \verifies\bigl[ - 3 |1001\ket + i |1010\ket + i |0101\ket - |0110\ket - 2 |0000\ket \bigr] $$

Considering the epistemic state at $(c^+, d^+, u^-, v^-)$, namely
$$ \bigl[2 i|1001\ket - |1010\ket + i |0110\ket - \sqrt 2 |0000\ket \bigr], $$
if the positron is found at detector $d^+$, then the electron must have followed the $w^-/u^-$ path. This corresponds to the fact that the only term not of the form $|\blank \,0\, \blank \blank \ket$ in the previous state is $|0110\ket$. Similarly, if the electron is found at $d^-$, this would imply that the positron has taken the $w^+/u^+$ arm of the interferometer.

But now, using Lorentz invariance, considering a reference frame $F^-$ in which the electron is found at $d^-$ before the positron passes the second beamsplitter. In that frame, the positron has to be in path $w^+/u^+$. Similarly, in the reference frame $F^+$ where the positron is found at $d^+$ before the electron has passed the second beamsplitter, the electron has to be in path $ w^-/u^-$. Consider now a third reference frame $F^=$ containing events $v^\pm$ and $u^\pm$. In that frame, both particle would be in their $u$ arms, which is impossible because they would have annihilated each other already. As a consequence, we would predict that it is not possible to find both particles at the $d$ detectors.

But quantum mechanics predicts that it is indeed possible to find both particles at the $d$ detectors (it is the term $- |0110\ket$ in the epistemic state at $(c^+, d^+, d^-, c^-)$) and actual experiments have confirmed these predictions \cite{LundeenSteinberg09:Hardy,Yokota09:Hardy}.

Let us now study these deductions in our formalism. First, it is clear that one cannot find both particles in arms $u^+$ and $u^-$, as
$$ u^+, u^- \verifies [11]^\bot \ \atfter \ v^+, v^- $$
Consider now the situation where the positron has actually been measured at $d^+$, which corresponds to the circuit $\mf H \cup \{ \star = \Mes_{[1]}(d^+)\}$ (it is not necessary to cut the circuit after $d^+$ since this s-node has not outgoing arrow). From the statement
$$ c^+, d^+, u^-, v^- \verifies\bigl[2 i|1001\ket - |1010\ket + i |0110\ket - \sqrt 2 |0000\ket \bigr], $$
we deduce
$$ \mf H \cup \{ \star = \Mes_{[1]}(d^+)\} \vdash c^+, \ \star \ , u^-, v^- \verifies [0110] $$
and, in particular,
$$ \mf H \cup \{ \star = \Mes_{[1]}(d^+)\} \vdash u^-\verifies [1] \ \atfter \ c^+, v^-, \star $$
so that it is not possible to find the electron in the $v^-$ arm, from which we deduce that it has to take the $u^-$ arm. Similarly,
$$ \mf H \cup \{ \star = \Mes_{[1]}(d^-)\} \vdash u^+\verifies [1] \ \atfter \ c^-, v^+, \star $$ 
Combining both circuits, by putting
$$ \mf H' = \mf H \cup \{ \star^+ = \Mes_{[1]}(d^+), \star^- = \Mes_{[1]}(d^-)\}, $$
one has
$$
\mf H' \vdash u^+\verifies [1] \ \atfter \ c^-, v^+, \star^- \quad \hbox{and} \quad \mf H' \vdash u^-\verifies [1] \ \atfter \ c^+, v^-, \star^+ $$
but it is not possible to deduce from this any verification statement of the form
$$ \mf H' \vdash u^+, u^- \verifies [11] \atfter \Delta $$
since there is no slice in $\mf H'$ containing both $\{u^+, v^+, c^-, \star^-\}$ and $\{u^-, v^-, c^+, \star^+\}$. There is thus no way to contradict the previous verification statement
$$u^+, u^- \verifies [11]^\bot$$
which would have entailed $u^+, u^- \verifies [11]^\bot \atfter \Delta$ for any suitable $\Delta$.

\ 

Let's express the same argument again in the following simpler setup, inspired from \cite{vaidman:on_hardy}: consider a system made of two particles $A$ and $B$ and suppose that they are prepared in a state
$$\mathopen|\Psi_{\!H}\ket = \frac 1 {\sqrt 3} \bigl(\mathopen|\uparrow\ket_{\!A} \mathopen| \uparrow\ket_{\!B} + \mathopen|\downarrow\ket_{\!A} \mathopen|\uparrow\ket_{\!B} + \mathopen|\uparrow\ket_{\!A} \mathopen|\downarrow\ket_{\!B}\bigr), $$
and consider the situation where particle $A$ is measured with outcome $\bigl[\mathopen|\uparrow\ket_{\!A} - \mathopen|\downarrow\ket_{\!A}\bigr]$ and $B$ with outcome $\bigl[\mathopen|\uparrow\ket_{\!B} - \mathopen|\downarrow\ket_{\!B}\bigr]$, as depicted below:

\hbox{} \hfill %
\begin{tikzpicture}[point/.style={circle,inner sep=0pt,outer sep = 0pt, minimum size=2pt,fill=black}]
	\matrix[row sep=4mm,column sep=5mm]{
		\node (a1) [point] {}; & \node[rectangle] (A) {$\bigl[\mathopen|\uparrow\ket_{\!A} - \mathopen|\downarrow\ket_{\!A}\bigr]$} ; & \node (a2) [point] {}; \\
		\node (b1) [point] {}; & \node[rectangle] (B) {$\bigl[\mathopen|\uparrow\ket_{\!B} - \mathopen|\downarrow\ket_{\!B}\bigr]$} ; & \node (b2) [point] {}; \\
	};
\node [above = -2pt of a1] {$A_1$} ;
\node [above = -2pt of b1] {$B_1$} ;
\node [above = -2pt of a2] {$A_2$} ;
\node [above = -2pt of b2] {$B_2$} ;
\draw (a1) -- (A) -- (a2) ;
\draw (b1) -- (B) -- (b2) ;
\draw [thick] (A.north west) -- (A.south west) ;
\draw [thick] (A.north east) -- (A.south east) ;
\draw [thick] (B.north west) -- (B.south west) ;
\draw [thick] (B.north east) -- (B.south east) ;
\draw[decorate,decoration={brace}, thick] ($(b1) + (-9pt, -4pt)$) -- ($(a1) + (-9pt, 4pt)$)
  node[midway, left] (bracket) {$\bigl[\Psi_{\!H}\bigr]\ $} ;
\end{tikzpicture} %
\hfill \hbox{}

We have $A_1, B_1 \verifies [\Psi_{\!H}]$, and $A_2 = \Mes\bigl(A_1, \bigl[\mathopen|\uparrow\ket_{\!A} - \mathopen|\downarrow\ket_{\!A}\bigr]\bigr)$ so that
$$ A_2, B_1 \verifies [\Psi_{\!H}] \sas \bigl(\bigl[\mathopen|\uparrow\ket_{\!A} - \mathopen|\downarrow\ket_{\!A}\bigr] \otimes \top\bigr) = \bigl[ \mathopen|\uparrow\ket_{\!A} - \mathopen|\downarrow\ket_{\!A}\bigr] \otimes \bigl[\mathopen| \downarrow \ket_{\!B} \bigr]$$
An analysis can be conducted as follows: in a reference frame where particle $A$ is measured before $B$, the latter is in state $\mathopen|\downarrow\ket_{\!B}$ before being measured. Similarly, in a reference frame where $B$ is measured $A$, the measurement of $B$ with outcome $\bigl[\mathopen|\uparrow\ket_{\!B} - \mathopen|\downarrow\ket_{\!B}\bigr]$ entails that $A$ is in state $\mathopen|\downarrow\ket_{\!A}$ before it is measured.

Thus, using Lorentz invariance, prior to any measurement, particle $A$ and $B$ both respectively verify $[\downarrow_A]$ at $A_1$ and $[\downarrow_B]$ at $B_1$, so that the joint system would be such that
$$ A_1, B_1 \verifies [\downarrow_A \downarrow_B]$$
as follows from the \emph{Tens} and \emph{Meet} rules. But then using the \emph{Meet} rule again, we would have
$$ A_1, B_1 \verifies \bot = [\Psi_H] \wedge [\downarrow_A \downarrow_B]$$
It is, however, not possible to derive such a result. Formally, using the \emph{Sas} rule, we obtain

\begin{prooftree}
\AxiomC{$A_1, B_1 \verifies [\Psi_{\!H}]$}
\AxiomC{$A_2 = \Mes\bigl(A_1, \bigl[\mathopen|\uparrow\ket_{\!A} - \mathopen|\downarrow\ket_{\!A}\bigr]\bigr)$}
\RightLabel{Sas}
\BinaryInfC{$A_2, B_1 \verifies \bigl[ \mathopen|\uparrow\ket_{\!A} - \mathopen|\downarrow\ket_{\!A}\bigr] \otimes \bigl[\mathopen| \downarrow \ket_{\!B} \bigr]$}
\end{prooftree}
Similarly, since $B_2 = \Mes\bigl(B_1, \bigl[\mathopen|\uparrow\ket_{\!B} - \mathopen|\downarrow\ket_{\!B}\bigr]\bigr)$, we also~have (modulo the correct permutations)
$$ A_1, B_2 \verifies [\Psi_{\!H}] \sas \bigl(\top \otimes \bigl[\mathopen|\uparrow\ket_{\!A} - \mathopen|\downarrow\ket_{\!A}\bigr] \bigr) = \bigl[ \mathopen| \downarrow_1 \ket \otimes \bigl(\mathopen|\uparrow\ket_{\!B} - \mathopen|\downarrow\ket_{\!B}\bigr) \bigr]$$
Now, using the \emph{Ord} rule, it follows from this that
$$ A_1, B_2 \verifies [\downarrow_A] \otimes \! \top \qquad \hbox{and} \qquad A_2, B_1 \verifies \top \otimes [\downarrow_B]$$
or, equivalently,
$$ A_1 \verifies [\downarrow_A] \atfter B_2 \qquad \hbox{and} \qquad B_1 \verifies [\downarrow_B] \atfter A_2 $$
which is dramatically different from having $A_1 \verifies [\downarrow_A]$ and $B_1 \verifies [\downarrow_B]$. We recall that here, the statements with the ``$\atfter$'' part means that in slice $A_1, B_2$ and more generally, following theorem~\ref{theo:at_reads_after}, in any slice containing $A_1$ \emph{and having $B_2$ in its strong past}, $A_1$ does verify~$[\downarrow_A]$. Similarly, in any slice containing $B_1$ and having $A_2$ in its strong past, $B_1$ verifies~$[\downarrow_B]$. However, there exists no slice verifying these two conditions:
\begin{enumerate}
	\item it contains both $A_1$ and $B_1$,
	\item it has both $A_2$ and $B_2$ in its strong past.
\end{enumerate}
as illustrated below\footnote{Using the notations from the analysis of circuit $\mf H$, reference frame $F^-$ corresponds to slice $[A_1, B_2]$, $F^+$ to $[A_2, B_1]$ and $F^=$ to $[A_1, B_1]$.}:

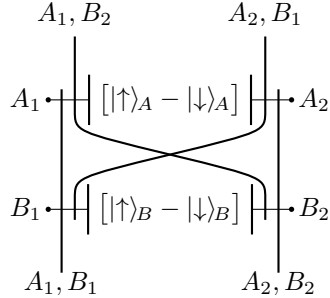
\begin{centering}
\begin{tikzpicture}[point/.style={circle,inner sep=0pt,outer sep = 0pt, minimum size=2pt,fill=black}]
	\matrix[row sep=8mm,column sep=5mm]{
		\node (a1) [point] {}; & \node[rectangle] (A) {$\bigl[\mathopen|\uparrow\ket_{\!A} - \mathopen|\downarrow\ket_{\!A}\bigr]$} ; & \node (a2) [point] {}; \\
		\node (b1) [point] {}; & \node[rectangle] (B) {$\bigl[\mathopen|\uparrow\ket_{\!B} - \mathopen|\downarrow\ket_{\!B}\bigr]$} ; & \node (b2) [point] {}; \\
	};
\node [left = -2pt of a1] {$A_1$} ;
\node [left = -2pt of b1] {$B_1$} ;
\node [right = -2pt of a2] {$A_2$} ;
\node [right = -2pt of b2] {$B_2$} ;
\draw (a1) -- (A) -- (a2) ;
\draw (b1) -- (B) -- (b2) ;
\draw [thick] ($(a1) + (5pt, 4pt)$) -- ($(b1) + (5pt, -24pt)$) node [below = -4pt] {$A_1,B_1$} ;
\draw [thick] ($(a1) + (10pt, 24pt)$) node [above] {$A_1, B_2$} -- ($(a1) + (10pt, -4pt)$) .. controls ($(a1) + (10pt, -12pt)$) .. ($0.5*(b2) + 0.5*(a1)$) .. controls ($(b2) + (-10pt, 12pt)$) .. ($(b2) + (-10pt, 4pt)$) -- ($(b2) + (-10pt, -4pt)$) ;
\draw [thick] ($(b1) + (10pt, -4pt)$) -- ($(b1) + (10pt, 4pt)$) .. controls ($(b1) + (10pt, 12pt)$) .. ($0.5*(a2) + 0.5*(b1)$) .. controls ($(a2) + (-10pt, -12pt)$) .. ($(a2) + (-10pt, -4pt)$) -- ($(a2) + (-10pt, 24pt)$) node [above] {$A_2, B_1$};
\draw [thick] ($(a2) + (-5pt, 4pt)$) -- ($(b2) + (-5pt, -24pt)$) node [below = -4pt] {$A_2, B_2$} ;
\draw [thick] (A.north west) -- (A.south west) ;
\draw [thick] (A.north east) -- (A.south east) ;
\draw [thick] (B.north west) -- (B.south west) ;
\draw [thick] (B.north east) -- (B.south east) ;
\end{tikzpicture}  \\
\end{centering}
\noindent so that it is not possible to find a slice $\Delta$ such~that
$$ A_1, B_1 \verifies [\downarrow_A] \otimes [\downarrow_B] \atfter \Delta.$$

What conclusions can be drawn from this analysis? We think that the main lesson is that Redhead's definition of element of reality should be slightly modified, by attaching the value of a physical quantity not to a time (or, more generally, to a spacetime event) but more generaly to one or more spacetime events, as embodied by the slices of our formalism. In that case, they become clearly Lorentz invariant, as we specify more accurately which hypersurface can be taken into account when considering a given element of reality.

It has also been objected that the use of the ``product'' and the ``and'' rules could be at the origin of the difficulties exemplified by Hardy's paradox \cite{Vaidman93:EltsReality,CohenHiley95:Reexaming,CohenHiley96:EltsReality,vaidman:on_hardy}. Here, this rule corresponds to the \emph{Meet} rule which has been shown to be correct. Again, this is only possible because whole slices are taken into consideration, and, modulo the application of the \emph{Tens} rule, the conjunction of two verification statement can only be defined if they apply to compatible slices.

\section{Conclusion} 

The logical formalism we have developed in this article started as an attempt to carefully define a set of rules for telling whether a given quantum circuit is possible, i.e.\ whether it represents a physical experimental setup and a set of measurement outcomes which can actually be obtained.

Through the notion of verification statement, we have seen that to each slice (i.e.\ what corresponds in the circuit formalism to a finite sets of spacelike separated events) one could associate subspaces of the corresponding Hilbert space and, in particular, a minimal one (w.r.t.\ inclusion), which we call the \emph{epistemic state} of the slice. Here, the adjective \emph{epistemic} refers to the fact that verifications statements are indeed defined in a purely epistemic way, as they constitue predictive statement regarding the possibility of obtain certain outcomes and are defined by only refering to the experimentally accessible information, namely previous measurement outcomes and the structure of the experimental setup.

As we have seen, this leads to a formulation of quantum mechanics where states (at least epistemic ones) are functions of slices (and, more generally, to spacelike hypersurfaces) rather that of space-time events. This is by no means a new idea, as such formulations can be traced back to Dirac, to Tomonaga 
and Schwinger, 
and more recent discussions have argued that this was indeed a necessity in order to have a Lorentz-invariant realistic interpretation of quantum mechanics \cite{Aharonov-Albert:IsUsualII,CohenHiley95:Reexaming}. However, the logical formalism which we have presented in this article, as defined in figure~\ref{fig:last_logic}, does indeed provide such a formulation, at least in the context of quantum circuits.

Let us, finally, remark that in this formalism, the basic element is constituted by verification statements which, we stress again, are purely epistemic. Yet, in many cases, it does accurately resemble what one would take for a quantum vector state (in particular when the epistemic state of a slice is a one-dimensional subspace). In our opinion, this should be interpreted as meaning that quantum vector state should, in general, be seen as particular types of verification statements, and hence should be seen as begin of epistemic nature. In order to obtain an ontic interpretation of quantum mechanics from our logical approach, it would be interesting to consider the models \cite{Marker:ModelTheory,Marker:Introduction,Hodges:Shorter} of our theory.

\bibliographystyle{alpha}

\begin{thebibliography}{BBC{\etalchar{+}}93}

\bibitem[AA84]{Aharonov-Albert:IsUsualII}
Yakir Aharonov and David Albert.
\newblock Is the usual notion of time evolution adequate for quantum-mechanical
  systems? {II}.~{R}elativistic considerations.
\newblock {\em Physical Review D}, 29(2), 1984.

\bibitem[BBC{\etalchar{+}}93]{Bennett93Teleportation}
Charles~H. Bennett, Gilles Brassard, Claude Cr\'epeau, Richard Jozsa, Asher
  Peres, and William~K. Wootters.
\newblock Teleporting an unknown quantum state via dual classical and
  {E}instein-{P}odolsky-{R}osen channels.
\newblock {\em Physical Review Letters}, 70(13):1895--1899, 1993.

\bibitem[Bru07]{Brunet07PLA}
Olivier Brunet.
\newblock A priori knowledge and the {K}ochen-{S}pecker theorem.
\newblock {\em Physical Letters A}, 365(1-2):39--43, May 2007.

\bibitem[Bru09]{Brunet:2009}
Olivier Brunet.
\newblock {P}artial {D}escription of {Q}uantum {S}tates.
\newblock {\em International Journal of Theoretical Physics}, 48(3), March
  2009.

\bibitem[Bru15]{Brunet:WBR}
Olivier Brunet.
\newblock Quantum measurements from a logical point of view.
\newblock In Chris Heunen, Peter Selinger, and Jamie Vicary, editors, {\em
  Proceedings 12th International Workshop on Quantum Physics and Logic}, volume
  195 of {\em EPTCS}, 2015.

\bibitem[CH95]{CohenHiley95:Reexaming}
O.~Cohen and B.~J. Hiley.
\newblock Reexamining the assumption that elements of reality can be lorentz
  invariant.
\newblock {\em Physical Review A}, 52(1), 1995.

\bibitem[CH96]{CohenHiley96:EltsReality}
O.~Cohen and B.~J. Hiley.
\newblock Elements of reality, lorentz invariance and the product rule.
\newblock {\em Foundations of Physics}, 26(1), 1996.

\bibitem[CN00]{NielsenChung2000Book}
Isaac~L. Chuang and Michael~A. Nielsen.
\newblock {\em Quantum Computation and Quantum Information}.
\newblock Cambridge, 2000.

\bibitem[CPP92]{CliftonPagonisPitowsky92}
Robert Clifton, Constantine Pagonis, and Itamar Pitowksy.
\newblock Relativity, {Q}uantum {M}echanics and {E}{P}{R}.
\newblock In {\em Philosophy of Science Association}, 1992.

\bibitem[EPR35]{Einstein35EPR}
Albert Einstein, Boris Podolsky, and Nathan Rosen.
\newblock Can quantum-mechanical description of physical reality be considered
  complete?
\newblock {\em Physical Review}, 47:777--780, 1935.

\bibitem[Fri09]{Fritz:Possibilistic}
Tobias Fritz.
\newblock {P}ossibilistic {P}hysics, October 2009.

\bibitem[GHSZ90]{GHZ90}
Daniel~M. Greenberger, Michael~A. Horne, Abner Shimony, and Anton Zeilinger.
\newblock {B}ell's {T}heorem without {I}nequalities.
\newblock {\em American Journal of Physics}, 58(12):1131 -- 1143, 1990.

\bibitem[Har92]{Hardy92}
Lucien Hardy.
\newblock Quantum mechanics, local realistic theories, and {L}orentz-invariant
  realistic theories.
\newblock {\em Physical Review Letters}, 68(20):2981 -- 2984, 1992.

\bibitem[Hod97]{Hodges:Shorter}
Wilfrid Hodges.
\newblock {\em A {S}hort {M}odel {T}heory}.
\newblock Cambridge University Press, 1997.

\bibitem[LS09]{LundeenSteinberg09:Hardy}
J.~S. Lundeen and A.~M. Steinberg.
\newblock Experimental joint weak measurement on a photon pair as a probe of
  hardy's paradox.
\newblock {\em Physical Review Letters}, 102(2), 2009.

\bibitem[Mar00]{Marker:Introduction}
David Marker.
\newblock Introduction to model theory.
\newblock In {\em Model Theory, Algebra and Geometry}, volume~39. MSRI
  Publications, 2000.

\bibitem[Mar02]{Marker:ModelTheory}
David Marker.
\newblock {\em Model Theory: An Introduction}, volume 217 of {\em Graduate
  Texts in Mathematics}.
\newblock Springer, 2002.

\bibitem[Per02]{Peres2002:QuantumTheory}
Asher Peres.
\newblock {\em Quantum Theory: Concepts and Methods}, volume~72 of {\em
  Fudamental Theories of Physics}.
\newblock Kluwer, 2002.

\bibitem[Red87]{Redhead:Realism}
Michael Redhead.
\newblock {\em Incompleteness, {N}onlocality and {R}ealism}.
\newblock Clarendon Press, 1987.

\bibitem[Vai93]{Vaidman93:EltsReality}
Lev Vaidman.
\newblock Elements of reality and the failure of the product rule.
\newblock In {\em SFMP}, 1993.

\bibitem[Vai97]{vaidman:on_hardy}
Lev Vaidman.
\newblock The analysis of {H}ardy's experiment revisited, March 1997.

\bibitem[YYKI09]{Yokota09:Hardy}
K.~Yokota, T.~Yamamoto, M.~Koashi, and N.~Imoto.
\newblock Direct observation of hardy's paradox by joint weak measurement with
  an entangled photon pair.
\newblock {\em New Journal of Physics}, 11(3), 2009.

\end{thebibliography}
\newcommand{\etalchar}[1]{$^{#1}$}

\end{document}